%% file: main-recursive-planarity.tex
\documentclass[11pt]{article}
\usepackage{ifthen}
\usepackage{amsmath, amssymb}

\usepackage{amsthm}

\usepackage{ifpdf}

\usepackage{fullpage}
\usepackage[utf8]{inputenc}

\usepackage{algorithm,algpseudocode}
\algtext*{EndWhile} % Removes the "end while" text
\algtext*{EndIf} % Removes the "end if" text
\algtext*{EndFor} % Removes the "end for" text
\usepackage{amsfonts}
\usepackage{mathrsfs}
\usepackage{comment}
\usepackage{mathtools}

\usepackage{enumitem}

\usepackage{xcolor}

\usepackage{mdframed}
\usepackage{tikz}

\usepackage{bbm}
\usepackage{dsfont}
\usepackage{hyperref}
\usepackage{subcaption}

\theoremstyle{plain}
\newtheorem{theorem}{Theorem}[section]
\newtheorem{lemma}[theorem]{Lemma}

\newtheorem{observation}[theorem]{Observation}

\theoremstyle{definition}

\newcommand{\poly}{\textnormal{poly}}

\newcommand{\pos}{\mathtt{pos}}

\newcommand{\triple}{\mathtt{TripleLog}}
\newcommand{\double}{\mathtt{DoubleLog}}
\newcommand{\hatpos}{\widehat{\pos}}
\newcommand{\iterated}{\mathtt{IteratedLog}}
\newcommand{\ygnote}[1]{{\color{red}[Yuval: #1]}}

\newcommand{\rounds}{\mathtt{rounds}}
\newcommand{\soundness}{\mathtt{soundErr}}
\newcommand{\completeness}{\mathtt{compErr}}
\newcommand{\size}{\mathtt{size}}
\renewcommand{\paragraph}{\smallskip\noindent\textbf}
\usepackage{cleveref}
\begin{document}
	
	\title{Distributed Interactive Proofs for Planarity with Log-Star Communication\thanks{This project is partially funded by the European Research Council (ERC) under the European Union’s Horizon 2020 research and innovation programme, grant agreement No.\ 949083}}
	\author{Yuval Gil\footnote{Reykjavik University. This research was done while the author was a postdoc at Weizmann Institute of Science. yuvalg@ru.is} \and Merav Parter \footnote{Weizmann Institute of Science. merav.parter@weizmann.ac.il}
	}
	\date{}
	
	\maketitle
	%%%%%%%%%%%%%%%%%%%%%%%%%%%%%%%%%%%%%%%%%%%%%%%%%%%%%%%%%%%%%%%%%%%%%%%%%%%%%%
	\input{abstract}
	%%%%%%%%%%%%%%%%%%%%%%%%%%%%%%%%%%%%%%%%%%%%%%%%%%%%%%%%%%%%%%%%%%%%%%%%%%%%%%
	\thispagestyle{empty}
	\clearpage
	\thispagestyle{empty}
	\tableofcontents
	\clearpage
	\setcounter{page}{1}
	%%%%%%%%%%%%%%%%%%%%%%%%%%%%%%%%%%%%%%%%%%%%%%%%%%%%%%%%%%%%%%%%%%%%%%%%%%%%%%
	\input{intro}

	%%%%%%%%%%%%%%%%%%%%%%%%%%%%%%%%%%%%%%%%%%%%%%%%%%%%%%%%%%%%%%%%%%%%%%%%%%%%%%		
	\input{preliminaries}

	%%%%%%%%%%%%%%%%%%%%%%%%%%%%%%%%%%%%%%%%%%%%%%%%%%%%%%%%%%%%%%%%%%%%%%%%%%%%%%	
	\input{overview}

	%%%%%%%%%%%%%%%%%%%%%%%%%%%%%%%%%%%%%%%%%%%%%%%%%%%%%%%%%%%%%%%%%%%%%%%%%%%%%%	
	\input{recursive-sorting}

	%%%%%%%%%%%%%%%%%%%%%%%%%%%%%%%%%%%%%%%%%%%%%%%%%%%%%%%%%%%%%%%%%%%%%%%%%%%%%%	

	%%%%%%%%%%%%%%%%%%%%%%%%%%%%%%%%%%%%%%%%%%%%%%%%%%%%%%%%%%%%%%%%%%%%%%%%%%%%%%
	\clearpage
	\bibliographystyle{alpha}
	% the mandatory bibstyle
	
	\bibliography{references}
	%%%%%%%%%%%%%%%%%%%%%%%%%%%%%%%%%%%%%%%%%%%%%%%%%%%%%%%%%%%%%%%%%%%%%%%%%%%%%%
	\clearpage
	\input{appendix}
\end{document}

%% file: abstract.tex
\begin{abstract}
	We provide new communication-efficient distributed interactive proofs for planarity. The notion of a \emph{distributed interactive proof (DIP)} was introduced by Kol, Oshman, and Saxena (PODC 2018). In a DIP, the \emph{prover} is a single centralized entity whose goal is to prove a certain claim regarding an input graph $G$. To do so, the prover communicates with a distributed \emph{verifier} that operates concurrently on all $n$ nodes of $G$. A DIP is measured by the amount of prover-verifier communication it requires. Namely, the goal is to design a DIP with a small number of interaction rounds and a small \emph{proof size}, i.e., a small amount of communication per round.

	In prior work, Naor, Parter, and Yogev (SODA 2020) presented a $3$-round DIP protocol for planarity with a proof size of $O(\log n)$. Later on, Feuilloley et al.\ (PODC 2020) showed that the same proof size can be accomplished with a non-interactive protocol. In a very recent work by Gil and Parter (DISC 2025), a $5$-round protocol with a proof size of $O(\log \log n)$ is presented for \emph{embedded planarity}, which is defined such that an embedding of the graph is given (e.g., each node holds a clockwise ordering of its incident edges) and the goal is to decide if it is a valid planar embedding. In addition, Gil and Parter presented a protocol with a proof size of $O(\log\log n+\log \Delta)$ for (non-embedded) planarity, where $\Delta$ is the maximum degree of the graph. 
	
	In this work, we design DIP protocols that significantly improve the communication bounds of Gil and Parter. Our main result is an $O(\log ^{*}n)$-round DIP protocol for embedded planarity and planarity with a proof size of $O(1)$ and $O(\lceil\log \Delta/\log ^{*}n\rceil)$, respectively. In fact, this result can be generalized as follows. For any $1\leq r\leq \log^{*}n$, there exists an $O(r)$-round protocol for embedded planarity and planarity with a proof size of $O(\log ^{(r)}n)$ and $O(\log ^{(r)}n+\log \Delta /r)$, respectively.\footnote{The function $\log ^{(r)}n$ is defined recursively such that $\log ^{(1)}n=\log n$ and $\log ^{(r)}n=\log (\log^{(r-1)}n)$.} As an important step towards our main result, we also provide a $3$-round DIP protocol for embedded planarity and planarity with a proof size of $O(\log \log n)$ and $O(\log \log n+\log \Delta)$, respectively. This improves the round complexity of the protocol of Gil and Parter from $5$ to $3$ while maintaining the same proof size.

	One of the tools that we develop in order to obtain the main result is a novel \emph{self-reduction} for a task in which two bitstrings are encoded in a distributed manner and we wish to test whether they are equal. Specifically, we show that solving this equality task can be reduced to solving a constant number of equality tasks on exponentially smaller instances. This self-reduction only requires a constant number of interaction rounds. We believe that this self-reduction could be of independent interest.

\end{abstract}

%% file: intro.tex
\section{Introduction}\label{section:introduction}
The concept of interactive proofs, introduced by Goldwasser, Micali, and Rackoff \cite{GoldwasserMR89} and Babai \cite{Babai85}, generalizes nondeterminism by allowing a computationally weak verifier to interact with a powerful (yet unreliable) prover. This paradigm has since become fundamental in various domains such as cryptography and complexity theory. 

In recent years, the increasing prevalence of distributed systems naturally leads to the question of how such proof systems can be adapted to scenarios where the computation is spread across a large-scale communication network. This motivated the introduction of \emph{distributed interactive proofs (DIPs)} by Kol, Oshman, and Saxena \cite{kol2018interactive}. In the DIP setting, the verifier consists of $n$ nodes connected by a communication graph $G$. The prover is a single centralized entity that can see the entire graph and generates its proof by means of a back-and-forth interaction with all nodes. The main complexity measures of a DIP protocol are the number of interaction rounds and the proof size, i.e., the size of the largest message sent during the protocol.

This paper focuses on the task of \emph{planarity}, i.e., deciding if a given graph is planar.
Due to the inherent global nature of planarity testing \cite{GhaffariH16}, it is only natural to explore the power that can be gained from interacting with a global prover in the DIP setting. The distributed verification of planarity has recently garnered significant attention. 
The first non-trivial DIP result for planarity was given by Naor, Parter, and Yogev \cite{NaorPY20} in the form of a $3$-round protocol with $O(\log n )$ proof size.\footnote{The protocol of \cite{NaorPY20} is part of a more general compiler that transforms an $O(n)$-time computation in the RAM model into a $3$-round DIP with a proof size of $O(\log n )$. The protocol for planarity then follows from plugging in a linear-time planarity testing algorithm (e.g., \cite{hopcroft1974efficient}).} Soon after, Feuilloley et al. \cite{FeuilloleyFMRRT21,FeuilloleyF0RRT23} showed that the same proof size can be achieved for the planarity task without interaction or randomization, i.e., they provided a \emph{proof labeling scheme} for planarity with a proof size of $O(\log n)$. This result is complemented by a matching $\Omega(\log n)$ lower bound for planarity even in graphs of maximum degree $3$. The same lower bound was later shown to apply also when the verifier is randomized \cite{GilP2025}.

Following the work of Feuilloley et al., the primary remaining question was whether a sub-logarithmic proof size is possible if interaction is allowed. In a very recent work \cite{GilP2025}, Gil and Parter provided a positive answer to this question. In particular, it is shown that there is a $5$-round DIP for planarity with a proof size of $O(\log \log n+\log \Delta)$ in graphs of maximum degree $\Delta$. In addition, a $5$-round DIP with a proof size of $O(\log \log n)$ is provided for embedded planarity, outerplanarity, series-parallel graphs, and graphs of treewidth $2$, where the embedded planarity task is defined such that a graph drawing is given in a distributed manner and the goal is to verify if it is a valid planar embedding. 

The protocols of \cite{GilP2025} rely on an efficient solution to the \emph{LR-sorting} task introduced in that paper. In LR-sorting, all nodes are embedded on an oriented path $H\subseteq G$ and the goal is for each node to distinguish between its left and right $G$-neighbors (see Section \ref{section:preliminaries} for the concrete definition). In the current paper, we provide protocols with improved round/communication complexity for LR-sorting in planar graphs which directly leads to improved protocols for all graph families considered in \cite{GilP2025}. In particular, a surprising conclusion that arises from these protocols is that the task of deciding planarity and related graph families can be solved by a DIP that uses only $O(\log ^{*}n)$ bits of communication per node (with the asterisk that for the planarity task, a $\log \Delta$ term is added to the total communication). This is in stark contrast to the $\Omega(\log n)$ bits necessary for the non-interactive case.

\subsection{Model}\label{section:model} 
We consider the notion of a \emph{distributed interactive proof (DIP)} as introduced in \cite{kol2018interactive}. In the DIP setting, instances are pairs $(G,\mathcal{I})$ consisting of a graph $G=(V,E)$ and a \emph{local input assignment} $\mathcal{I}:V\to\{0,1\}^{*}$. The instances are taken from some universe $\mathcal{U}$ and the goal is to distinguish between \emph{yes-instances} that come from a yes-family $\mathcal{F}_{Y}\subset \mathcal{U}$ and \emph{no-instances} that come from a no-family $\mathcal{F}_{N}=\mathcal{U}-\mathcal{F}_{Y}$. A DIP is defined as an interactive protocol between a distributed \emph{verifier} operating concurrently at all nodes of $G$ and a centralized \emph{prover} which has access to the entire instance.  

In a DIP protocol, the prover and verifier communicate back and forth in \emph{rounds}. In each prover-interaction round $i$, the prover communicates with the verifier at each node $v\in V$ by sending a message $\mu_{v}^{i}\in \{0,1\}^{*}$. We sometimes refer to the messages sent by the prover as \emph{labels}. In a verifier-interaction round $i$, the verifier at each node $v\in V$ communicates with the prover by drawing and sending a random bitstring $\rho_{v}^{i}\in \{0,1\}^{*}$. Notice that in particular, the verifier cannot hide any random bits from the prover, i.e., the protocols are \emph{public-coin}. 

The interaction ends with a prover-interaction round, after which the verifier at each node $v\in V$ computes a local yes/no output based on: (1) the random bitstrings $\rho_{v}^{i}$ drawn by $v$ throughout the protocol; (2) the labels $\mu_{v}^{i}$ assigned to $v$ by the prover throughout the protocol; and (3) the labels $\mu_{u}^{i}$ assigned to $v$'s neighbors $u\in N(v)$ by the prover throughout the protocol. We say that the verifier \emph{accepts} the instance if all nodes output `yes', and that the verifier \emph{rejects} the instance if at least one node outputs `no'.

As standard, the correctness of a proof system is defined by \emph{completeness} and \emph{soundness} requirements. The completeness requirement asks that if $(G,\mathcal{I})\in \mathcal{F}_{Y}$, then there exists an \emph{honest} prover causing the verifier to accept the instance; whereas the soundness requirement asks that if $(G,\mathcal{I})\in\mathcal{F}_{N}$, then the verifier rejects the instance for any prover. In the DIP setting, the correctness requirements are relaxed to allow a probabilistic error. The complexity of a DIP protocol is measured by means of its communication bounds. Namely, the objective is to design protocols with a small number of interaction rounds and a small \emph{proof size} which is defined as the maximum size of a message exchanged between the honest prover and the verifier during the protocol. For a DIP protocol $\Pi$, let $\rounds(\Pi)$ denote the number of interaction rounds in $\Pi$; let $\size(\Pi)$ denote the proof size of $\Pi$; let $\completeness(\Pi)$ denote the completeness error of $\Pi$; and let $\soundness(\Pi)$ denote the soundness error of $\Pi$. We say that $\Pi$ has \emph{perfect completeness} if $\completeness(\Pi)=0$. 
%We denote by $\mathsf{DIP}_{\epsilon,\delta}(r,s) $ the class of decision problems for which there exists a DIP protocol $\Pi$ such that $\rounds(\Pi)=r$, $\size(\Pi)=s$, $\completeness(\Pi)=\epsilon$, and $\soundness(\Pi)=\delta$.

\subsection{Our Results}\label{section:results}
We present improved DIP protocols for deciding if a given graph is planar. In that context, we consider two tasks, namely \emph{planarity} and \emph{embedded planarity}. Embedded planarity is defined such that each node receives as local input a clockwise orientation of its edges, and the goal is to decide whether these orientations induce a valid planar embedding. In planarity, the nodes do not receive any local input and the goal is to decide if the given graph is planar. %The main results are stated in the following theorems.
We start by reducing the number of rounds obtained in \cite{GilP2025} from 5 to 3, while keeping the same proof size. 
\begin{theorem}\label{theorem:results-three-planarity}
	There exist $3$-round DIP protocols for embedded planarity and planarity with proof sizes of $O(\log \log n)$ and $O(\log \log n+\log \Delta)$, respectively. Both protocols have perfect completeness and a soundness error of $O(1/\log n)$.\footnote{In all our protocols, a soundness error of $\epsilon$ can be reduced to an arbitrarily small $\epsilon^{c}$ by means of standard parallel $c$-repetition.}
\end{theorem}

\noindent Our key result is an $O(\log ^{*}n)$-round protocol with constant proof size (in each round). This bound is achieved by a new technique of self-reduction for distributed equality tasks, which should be of independent interest (see Section \ref{section:self-reduction} for the details of the self-reduction).

\begin{theorem}\label{theorem:results-logstar-planarity}
	There exist $O(\log ^{*}n)$-round DIP protocols for embedded planarity and planarity with proof sizes of $O(1)$ and $O(\lceil\log \Delta/\log ^{*}n\rceil)$, respectively. Both protocols have perfect completeness and a soundness error of $1/4$.
\end{theorem}

\noindent In fact, Theorem \ref{theorem:results-logstar-planarity} is a special case of a more general trade-off between rounds and proof size. The details of this trade-off are stated in the following theorem.
\begin{theorem}\label{theorem:results-trade-off}
	For any $1\leq d\leq \log^{*}n$, there exist $O(d)$-round DIP protocols for embedded planarity and planarity with proof sizes of $O(\log^{(d+1)}n)$ and $O(\log^{(d+1)}n+\log \Delta/d)$, respectively. Both protocols have perfect completeness and a soundness error of $O(1/\log^{(d)}n)$.
\end{theorem}

In addition to the above results, we are able to obtain protocols with similar communication and error bounds for outerplanar graphs, series-parallel graphs, and graphs of treewidth $2$.
\begin{theorem}\label{theorem:results-outerplanar}
	There exist DIP protocols for outerplanar graphs, series-parallel graphs, and graphs of treewidth $2$ with the same communication and correctness guarantees as the protocols for embedded planarity stated in Theorems \ref{theorem:results-three-planarity}, \ref{theorem:results-logstar-planarity}, and \ref{theorem:results-trade-off}. 
\end{theorem}

%\mpnote{I added this section.}
\paragraph{Open Problems.} 
The most intriguing open problem is in closing the gap between the tasks of embedded planarity and planarity. In the current solution for planarity, the prover provides each node with a clockwise orientation of its neighbors in a planar embedding, thus reducing the problem to embedded planarity. As shown in \cite{GilP2025}, the clockwise orientation can be encoded using $O(\log \Delta)$ bits for each node. By a standard counting argument, it is also easy to show that $\Omega(\log \Delta)$ bits are necessary for each node to learn its clockwise orientation. 

As the vast majority of planar graph algorithms (centralized and distributed) are based on precomputing a planar embedding, intuitively it is natural to think that providing an embedding would be necessary also in the DIP setting. This intuition fails, however, as exemplified by the notion of \emph{distributed zero-knowledge(DZK)} proofs introduced by Bick, Kol and Oshman \cite{BickKO22}. A DZK proof is a DIP that in addition to the standard correctness requirements, asks that the verifier ``learns nothing'' (including, e.g., a planar embedding) from the interaction. We note that in particular, \cite{BickKO22} provides a generic result that transforms any proof labeling scheme to a DZK proof. In the context of planarity, combining this general transformation with the proof labeling scheme of \cite{FeuilloleyFMRRT21} results in a DZK proof with a proof size of $\poly(\Delta,\log n)$. While the large proof size of the DZK proof makes it impractical for our use, its existence demonstrates that providing an embedding may not be essential for verifying planarity. 
%Our solution for planarity provides each node with the clockwise orientation of its neighbors. This leads to the extra additive term of $O(\log \Delta)$ bits in the final label size. It is also easy to show that the encoding of the embedding must require $\Omega(\log \Delta)$ bits, in the worst case. As the vast majority of planar algorithms (centralized and distributed) are based on a precomputation of the planar embedding, this additive term seems to be necessary also for verifying planarity. 
%This intuition fails, however, due to the existence of distributed zero-knowledge proofs as obtained by Bick, Kol and Oshman \cite{BickKO22}. In particular, the result of \cite{BickKO22} implies a DIP that allows nodes to verify if the graph is planar without revealing any other information (e.g., the embedding) on the graph. Unfortunately, the label size of such proofs is very large, possibly linear in the number of edges. 

Another interesting open problem is in optimizing the label size of a constant-round protocol for embedded planarity. While we show that it is possible to obtain a DIP that uses only a total of $O(\log^* n)$ bits of communication per node, this also requires $\Theta (\log ^{*}n)$ interaction rounds. It is open whether a similar communication bound can be achieved while restricting the DIP to a constant number of rounds.

%% file: preliminaries.tex
%\ygnote{Possibly paragraph explaining why previous papers (NPY,GP) don't accomplish our results.}
\section{Preliminaries}\label{section:preliminaries}

\paragraph{LR-sorting.} 
The \emph{left-right sorting (LR-sorting)} task was introduced recently in \cite{GilP2025}. In LR-sorting, instances are pairs $(G,H)$ consisting of a directed graph $G=(V,E)$ and a directed Hamiltonian path $H$ of $G$. The directed graph $G$ is given such that each node can distinguish between its incoming and outgoing edges. The path $H$ is given such that each node $v\in V$ knows its incident edges in $H$. The goal in LR-sorting is to decide if $u$ appears before $v$ in $H$ for every directed edge $(u,v)\in E-H$. As noted in \cite{GilP2025}, an equivalent formulation of LR-sorting is the task of deciding if $G$ is a DAG (in which case the $H$-ordering is the unique topological sorting of $G$). See Figure \ref{figure:lr-sorting} for an example of an LR-sorting instance.

In \cite{GilP2025}, it is shown that the LR-sorting task can play an important role in obtaining protocols for planarity and planarity-related tasks. In particular, the next lemma follows from \cite[Lemmas 7.1,7.2]{GilP2025}.
\begin{lemma}[\cite{GilP2025}]\label{lemma:sorting-to-planarity}
	Suppose that $\Pi$ is a DIP protocol for LR-sorting in planar graphs. Then, there exists a DIP protocol $\Pi'$ for embedded planarity such that $\rounds(\Pi')=\max\{3,\rounds(\Pi)\}$; $\size(\Pi')=O(\size(\Pi))$; $\completeness(\Pi)'=\completeness(\Pi)$; and $\soundness(\Pi')\leq \soundness(\Pi)+2^{-\size(\Pi)}$. In addition, there is a protocol for planarity with the same rounds and correctness guarantees as $\Pi'$ and a proof size of $O(\size(\Pi)+\log \Delta/\rounds(\Pi'))$.
\end{lemma}
\noindent Regarding the added $\log \Delta/\rounds(\Pi')$ to the proof size in the planarity task, we note that \cite[Lemma 7.2]{GilP2025} shows that planarity can be reduced to embedded planarity while incurring an additive $O(\log \Delta)$ factor to the \emph{total} communication of the prover with each node in the protocol. Therefore, dividing this added communication uniformly over the rounds lead to the stated proof size.

In addition to the above reductions from (embedded) planarity to LR-sorting, \cite{GilP2025} presents reductions with similar communication and correctness guarantees from outerplanar graphs, series-parallel graphs, and graphs of treewidth $2$.
\begin{lemma}\label{lemma:sorting-to-outerplanarity}
	Given a DIP protocol $\Pi$ for LR-sorting in planar graphs, there exist DIP protocols for outerplanar graphs, series-parallel graphs, and graphs of treewidth $2$ with the same communication and correctness guarantees as the protocol $\Pi'$ stated in Lemma \ref{lemma:sorting-to-planarity} for embedded planarity.
\end{lemma}
%\ygnote{Maybe change to one lemma}
\begin{figure}
	\centering
	\includegraphics[width=\textwidth]{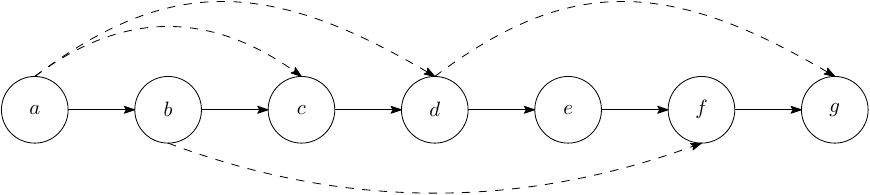}
	\caption{A valid instance of LR-sorting. The edges of the Hamiltonian path $H$ are solid and the edges of $E-H$ are dashed.}
	\label{figure:lr-sorting}
\end{figure} 
\paragraph{Enabling edge-labels in planar graphs.}
Recall that in DIP protocols, the prover interacts by assigning labels to the nodes. In Section \ref{section:sorting}, we present protocols under the stronger assumption that the prover can also assign labels to the edges (such that both endpoints of an edge can see its label). As shown in \cite[Lemma 2.4]{GilP2025}, edge-labels can be simulated by node-labels in planar graphs while incurring only a constant overhead to the proof size. That is, the next lemma follows directly from the construction described in \cite[Lemma 2.4]{GilP2025}.
\begin{lemma}\label{lemma:edge-labels}
	Let $\varphi$ be a decision problem defined over a family of planar graph instances and suppose that $\Pi$ is a protocol for $\varphi$ in which the prover assigns labels to both nodes and edges. Then, there exists a protocol $\Pi'$ for $\varphi$ such that: (1) the prover in $\Pi'$ only assigns node-labels; (2) $\size(\Pi')=O(\size(\Pi))$; and (3) $\Pi'$ has the same round-complexity and the same completeness and soundness errors as $\Pi$.
\end{lemma}

\subsection{Bitstrings and Arithmetic Operations on a Directed Path}\label{section:arithmetic}
%\paragraph{Verifying arithmetic operations in paths.}
In this section, we consider distributed tasks defined on oriented paths. Throughout the paper, when encountering an oriented path, it will be convenient to think of it as a straight line from left to right. For an oriented path $P$ consisting of $\ell$ nodes, we represent a bitstring $\alpha=\alpha[1]\alpha[2]\dots \alpha[\ell]$ in a distributed manner by assigning the bit $\alpha[i]$ to the $i$-th leftmost node in $P$. Going forward, when we say that a bitstring of length $\ell$ is written on $P$, we assume that the mechanism described above is used unless stated otherwise. Notice that in this representation, each node $v\in P$ only receives a single bit and in particular, $v$ is not aware of the index $i$ that reflects its position in $P$. We denote by $\alpha(v)$ the bit of $\alpha$ assigned to $v$. We often identify $\alpha$ as an integer from the range $\{0,\dots, 2^{\ell}-1\}$ according to its binary representation where $\alpha[i]$ represents the $i$-th most significant bit. 

We note that in the course of this paper, there may be some cases in which the bitstring $\alpha$ which is encoded on $P$ is of length $\ell'<\ell$. Unless stated otherwise, we assume that in these cases $\alpha$ is written on the $\ell'$ rightmost nodes of $P$ and that leading zeros are assigned to the $\ell-\ell'$ leftmost nodes of $P$. We now define some useful distributed tasks on bitstrings and establish protocols for them. All missing proofs of the current section are deferred to Appendix \ref{section:missing-proofs-prelims}.

\paragraph{Greater than.} In the \emph{greater than (GT)} task, an ordered pair $(\alpha,\beta)$ of integers $\alpha,\beta\in \{0,\dots ,2^{\ell}-1\}$ is written on the path $P$ and the goal is to decide if $\alpha>\beta$. The following lemma establishes the existence of a proof labeling scheme for GT.\footnote{Recall that proof labeling scheme refers to a non-interactive and non-randomized protocol \cite{KormanKP10}.} 
%The proof of the lemma is given in Appendix \ref{section:missing-proofs-prelims}.
\begin{lemma}\label{lemma:greater-than}
	There is a proof labeling scheme for GT with a proof size of $1$. 
\end{lemma}

\paragraph{Addition.} In the \emph{addition} task, an ordered triple $(\alpha,\beta,\gamma)$ of integers $\alpha,\beta,\gamma\in\{0,\dots ,2^{\ell}-1\}$ is written on the path $P$ and the goal is to decide if $\alpha+\beta=\gamma$. 
%The proof of the following Lemma is deferred to Appendix \ref{section:missing-proofs-prelims}.
\begin{lemma}\label{lemma:addition}
	There is a proof labeling scheme for addition with a proof size of $1$.
\end{lemma}

We also define a related \emph{modular addition} task as follows. A number $N\in \{0,\dots, 2^{\ell}-1\}$ is written on the path $P$ along with an ordered triple $(\alpha,\beta,\gamma)$ of integers $\alpha,\beta,\gamma\in \{0,\dots ,N-1\}$, and the goal is to decide if $\alpha+\beta\equiv\gamma\bmod N$. Due to the range of $\alpha,\beta,\gamma$, a protocol for modular addition is straightforward from the addition protocol since the prover only needs to prove that $\alpha+\beta=\gamma$ in the case that $\alpha+\beta<N$, and $\alpha+\beta=N+\gamma$ otherwise. Thus, we get the following. 
\begin{lemma}\label{lemma:addition-mod-p}
	There is a proof labeling scheme for modular addition with a proof size of $O(1)$.
\end{lemma}
\paragraph{Equality.}
In the \emph{equality} task, an ordered pair $(P,P')$ of node-disjoint paths is given such that each node $v\in P\cup P'$ knows whether it is in $P$ or $P'$. Both paths consist of $\ell$ nodes and are connected by some edge $e^{*}=(u^{*},v^{*})\in P\times P'$ which is marked such that both $u^{*}$ and $v^{*}$ are able to distinguish $e^{*}$ from their other incident edges. Two bitstrings $\alpha,\alpha'\in \{0,1\}^{\ell}$ are written on $P$ and $P'$, respectively, and the goal of the equality task is to decide if $\alpha=\alpha'$. See Figure \ref{figure:equality} for an example of an equality instance.

Towards describing a protocol for equality, we identify any bitstring $s\in\{0,1\}^{\ell}$ with the univariate polynomial $\Phi(x;s)=\sum_{i=1}^{\ell}s[i]\cdot x^{i}$. For a prime number $q> \ell$ and value $x\in \{0,\dots ,q-1\}$, let us denote by $\Phi_{q}(x;s)$ the evaluation of $\Phi(x;s)$ computed over the field $\mathbb{F}_{q}$.\footnote{Recall that $\mathbb{F}_{q}$ is the field whose elements are in $\{0,\dots, q-1\}$ and operations are done modulo $q$} The following lemma is an application of the well-known polynomial identity lemma (a.k.a.\ the Schwartz-Zippel lemma \cite{Schwartz80,Zippel79}) and is frequently used in various domains (e.g., communication complexity \cite{Kushilevitz-Nisan97}).
\begin{lemma}\label{lemma:poly-identity}
	Let $\alpha,\alpha'\in \{0,1\}^{\ell}$ such that $\alpha\neq \alpha'$ and let $q$ be a prime number such that $q>\ell$. Then, it holds that $\Pr_{r\in \{0,\dots , q-1\}}[\Phi_{q}(r;\alpha)=\Phi_{q}(r;\alpha')]\leq \ell/q$. 
\end{lemma} 

Based on Lemma \ref{lemma:poly-identity}, we can devise a protocol for equality with the following properties.
\begin{lemma}\label{lemma:equality}
	There is a $2$-round protocol for equality on instances $\alpha,\alpha'\in\{0,1\}^{\ell}$ with a proof size of $O(\log \ell)$, perfect completeness, and a soundness error of at most $1/\ell$.
\end{lemma}

\begin{figure}
	\centering
	\includegraphics[width=\textwidth]{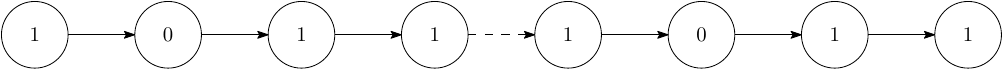}
	\caption{A valid instance of equality. Here, $e^{*}$ is depicted by the dashed edge which connects between $P$ on the left and $P'$ on the right. The encoded bitstrings are $\alpha=\alpha'=1011$.}
	\label{figure:equality}
\end{figure}

\paragraph{Multiplication.}
In the \emph{multiplication} task, an ordered triple $(\alpha,\beta,\gamma)$ of integers $\alpha,\beta,\gamma\in \{0,\dots ,2^{\ell'}-1\}$ is written on the $\ell'\leq \ell$ leftmost nodes of the path $P$ and the goal is to decide if $\alpha\cdot\beta=\gamma$. Notice that we allow for lengths $\ell'<\ell=|P|$. This becomes useful in the following lemma where we establish a reduction from multiplication to equality. The high-level idea is to have the prover encode the steps of a standard multiplication algorithm and use the equality protocol to verify that all the intermediate values are encoded correctly (refer to Appendix \ref{section:missing-proofs-prelims} for the full details).
\begin{lemma}\label{lemma:multiplication-to-equality}
	Let $\mathsf{EQ}_{k}$ be a protocol with perfect completeness for equality on instances of size $k$. Given a path $P$ on $\ell$ nodes and a multiplication instance $\alpha,\beta,\gamma\in \{0,\dots ,2^{\ell'}-1\}$ where $\ell'\leq \sqrt{\ell/2}$, there exists a protocol $\Pi$ that determines if $\alpha\cdot \beta=\gamma$ such that $\rounds (\Pi)=\rounds(\mathsf{EQ}_{2\ell'})+1-(\rounds(\mathsf{EQ}_{2\ell'})\bmod 2)$, $\size(\Pi)=O(\size(\mathsf{EQ}_{2\ell'}))$, $\completeness(\Pi)=0$, and $\soundness(\Pi)=\soundness(\mathsf{EQ}_{2\ell'})$.
\end{lemma}

Analogously to the case in the addition task, we define \emph{modular multiplication}. A number $N\in \{0,\dots, 2^{\ell'}-1\}$ is written on the path $P$ along with an ordered triple $(\alpha,\beta,\gamma)$ of integers $\alpha,\beta,\gamma\in \{0,\dots ,N-1\}$, and the goal is to decide if $\alpha\cdot\beta\equiv\gamma\bmod N$. By a similar reasoning to the addition case, we get the following lemma.
\begin{lemma}\label{lemma:multiplication-mod-p}
	Let $\ell\geq \ell'>0$ be two integers and suppose that $\Pi$ is a protocol with perfect completeness for multiplication on instances $\alpha,\beta,\gamma\in \{0,\dots ,2^{\ell'}-1\}$ encoded on a path $P$ of size $\ell$. Then, there exists a protocol $\Pi'$ for modular multiplication on instances $N\in \{0,\dots, 2^{\ell'}-1\}$, $\alpha,\beta,\gamma\in \{0,\dots ,N-1\}$ encoded on a path $P$ of size $\ell$. The protocol $\Pi'$ satisfies $\rounds(\Pi')=\rounds(\Pi)$, $\size(\Pi')=O(\size(\Pi))$, $\completeness(\Pi')=\completeness(\Pi)=0$, and $\soundness(\Pi')=\soundness(\Pi)$. 
\end{lemma}

%% file: overview.tex
\section{Technical Overview}
As demonstrated in \cite{GilP2025}, a technical barrier to obtaining efficient protocols for planarity is the LR-sorting task (see Lemma \ref{lemma:sorting-to-planarity}). Therefore, in Section \ref{section:sorting} we provide new protocols for LR-sorting in planar graphs. We now provide a technical overview of these protocols and the tools developed in order to obtain them. 

\paragraph{The $\double$ protocol.} In Section \ref{section:double-log}, we present the $\double$ protocol which solves LR-sorting in planar graphs in $3$ interaction rounds and a proof size of $O(\log \log n)$. Let $(G,H)$ be an instance for LR-sorting consisting of a directed planar graph $G=(V,E)$ and a Hamiltonian path $H$ of $G$. The protocol begins similarly to the one presented in \cite{GilP2025}: the prover partitions $H$ into \emph{blocks} of size $\log n$, where each block is made up of consecutive nodes in $H$. This block-partition naturally divides the edges of $G$ into two sets, namely \emph{inner-block} edges whose endpoints lie in the same block, and \emph{outer-block} edges whose endpoints lie in distinct blocks. The prover then seeks to prove the validity of the inner-block and outer-block edges w.r.t.\ LR-sorting. That is, for every inner-block/outer-block edge $(u,v)\in E$, the prover seeks to prove that $u$ appears before $v$ in $H$. 

In this section, we will focus on the validity of the outer-block edges. Let us assume for simplicity that for each outer-block edge $e=(u,v)\in E$, both endpoints know that $e$ is an outer-block edge. Moreover, assume that each block $b$ is given a value $\pos(b)$ reflecting its \emph{position} in the ordering induced by $H$. The value $\pos(b)$ is given to $b$ in a distributed manner (as described in Section \ref{section:arithmetic}) such that each node $v\in b$ receives a single bit from $\pos(b)$. For an outer-block edge $e=(u,v)\in E$, let $b_{u}$ denote $u$'s block and let $b_{v}$ denote $v$'s block. To prove the validity of $e$, the prover needs to prove that $\pos(b_{u})<\pos(b_{v})$.

Let us first consider the case where $e$ is the \emph{only} outer-block edge incident on $b_{u}$. The idea for proving $e$'s validity is as follows. First, the prover assigns a value $z(b_{u})=\pos(b_{v})$ to the block $b_{u}$ and provides proof that $z(b_{u})>\pos(b_{u})$ based on the scheme for the GT task described in Lemma \ref{lemma:greater-than}. Then, it is left for the prover to provide proof that indeed $z(b_{u})$ is equal to $\pos(b_{v})$. Notice that this is an equality task. Therefore, it can be solved by executing the protocol of Lemma \ref{lemma:equality} on $b_{u}$ and $b_{v}$. Since the blocks are of size $\log n$, we get that the proof size is $O(\log \log n)$. 

Moving on to the general case, the challenge is that each block $b$ may have many incident outer-block edges. Assigning $b$ with a bitstring for each such edge would result in a large proof size. The main observation is that based on some standard properties of planar graphs, it is possible to assign each block with only a constant number of equality instances such that every outer-block edge is ``covered'' by an equality instance. In other words, the task of verifying the validity of outer-block edges reduces to solving a constant number of equality tasks in each block. Therefore, towards obtaining protocols with $o(\log \log n)$ total communication, we focus on reducing the proof size associated with equality.

\paragraph{The self-reduction for equality.}
%\ygnote{maybe add some details}
In Section \ref{section:self-reduction}, we present a self-reduction for a variation of the equality task referred to as $(g,\ell)$-equality. The $(g,\ell)$-equality task differs from standard equality by the way in which the bitstrings $\alpha,\alpha'$ are encoded on their respective paths $P,P'$. Specifically, in $(g,\ell)$-equality the bitstrings are of length $\ell$ and they are encoded on paths of size $g\cdot \ell$ such that every pair of consecutive bits are assigned to nodes at path-distance $g$. We now sketch the idea of how $(g=\Theta(\log^{2} \ell),\ell )$-equality can be reduced to a constant number of (standard) equality tasks on instances of size $O(\log \ell)$. As we will see, the role of the \emph{gap} component $g$ is to provide space for the prover to encode bitstrings as well as proofs for arithmetic operations (particularly multiplication) without overloading the nodes' labels. This idea can be easily extended to obtain the full statement presented in Lemma \ref{lemma:self-reduction}. 
%\mpnote{Here we should provide a quick explanation for this variant of equality, i.e., that the gap is introduced when computing multiplication of two distributed stings on a path.}

%\mpnote{I don't see where the gap $g$ is used in the explanation below.}

The intuitive idea behind the self-reduction is the following. Let $q$ be a prime number in the range $[\ell^{2},2\ell^{2}]$ and let $e^{*}=(u^{*},v^{*})\in P\times P'$ be the edge guaranteed from the equality definition. Based on Lemma \ref{lemma:poly-identity}, equality between $\alpha$ and $\alpha'$ can be tested by checking that $\Phi_{q}(r;\alpha)=\Phi_{q}(r;\alpha')$ for a uniformly drawn $r\in \{0,\dots, q-1\}$. So, if the prover encodes (in a distributed manner) the value $\Phi_{q}(r;\alpha)$ on a sequence of $O(\log q)=O(\log \ell)$ nodes that include $u^{*}$, and $\Phi_{q}(r;\alpha')$ on a sequence of $O(\log q)=O(\log \ell)$ nodes that include $v^{*}$, then verifying the equality between the bitstrings is precisely an equality task. Of course, this equality alone is not enough since a malicious prover can simply \emph{lie} about the values $\Phi_{q}(r;\alpha),\Phi_{q}(r;\alpha')$. It is therefore the main objective to have the prover provide proof that the polynomial values $\Phi_{q}(r;\alpha)$ and $\Phi_{q}(r;\alpha')$ were computed correctly. Interestingly, it turns out that this \emph{polynomial evaluation} verification task can also be reduced to equality between bitstrings of length $O(\log \ell)$. 

We consider the polynomial evaluation task w.r.t.\ the value $\Phi_{q}(r;\alpha)$; the case for $\Phi_{q}(r;\alpha')$ is symmetrical. For each $i\in [\ell]$, the prover assigns the value $\Phi_{q}(r;\alpha[1\dots i])$ to a \emph{segment} of $O(\log \ell)$ consecutive nodes starting from the node to which $\alpha[i]$ is encoded. The prover then seeks to prove the correctness of the assignments. Notice that if the prover manages to prove all $\Phi_{q}(r;\alpha[1\dots i])$ assignments, then in particular it proves the correctness of $\Phi_{q}(r;\alpha[1\dots \ell])=\Phi_{q}(r;\alpha)$ (which is the objective of the polynomial evaluation task).

Consider the segment $S_{i}$ that is associated with $\alpha[i]$. Towards proving the correctness of the $\Phi_{q}(r;\alpha[1\dots i])$ assignment, the prover provides an encoding of $q$, $r$, $\Phi_{q}(r;\alpha[1\dots i-1])$, $r^{i-1}\bmod q$, and $r^{i}\bmod q$. Notice that if the encodings are correct, then it is only left for the verifier to check that $\Phi_{q}(r;\alpha[1\dots i])=\Phi_{q}(r;\alpha[1\dots i-1])$ if $\alpha[i]=0$; and  $\Phi_{q}(r;\alpha[1\dots i])=[\Phi_{q}(r;\alpha[1\dots i-1])+(r^{i}\bmod q)]\bmod q$ otherwise. The latter check can be done by an application of the modular addition protocol described in Lemma \ref{lemma:addition-mod-p}. 

To check the correctness of $r^{i}\bmod q$ (while assuming correctness of the other assignments), the verifier needs to check that $r^{i}\bmod q=[r\cdot (r^{i-1}\bmod q)]\bmod q$. Notice that this is the modular multiplication task which can be reduced to equality based on Lemmas \ref{lemma:multiplication-mod-p} and \ref{lemma:multiplication-to-equality}. Here, due to the selection of $g=\Theta(\log^{2}\ell)$ and the parameters of Lemma \ref{lemma:multiplication-to-equality}, every node participates in the proof of only one modular multiplication (and thus, is associated with a constant number of equality instances that are created as a consequence). Finally, verifying the correctness of the assignments of $q$, $r$, $\Phi_{q}(r;\alpha[1\dots i-1])$, and $r^{i-1}\bmod q$ can be done by comparing them with the corresponding assignments to the segment $S_{i-1}$ associated with $\alpha[i-1]$. A small technical detail here is that $S_{i-1}$ and  $S_{i}$ are not necessarily adjacent, thus these are not equality instances per se. Without getting into the concrete details here, we note that this technicality is easy to mend. The final outcome of the self-reduction is a partition of $P$ and $P'$ into segments of length $O(\log \ell)$ such that each segment is associated with a constant number of equality instances of size $O(\log \ell)$.

\paragraph{The $\iterated$ protocol.} 
In Section \ref{section:recursion}, we present the $\iterated$ protocol which solves LR-sorting in planar graphs in $O(\log ^{*}n)$ interaction rounds and a proof size of $O(1)$. In fact, as we state in Theorem \ref{theorem:general-iterated-log}, $\iterated$ can be modified to obtain a more general trade-off between interaction rounds and proof size. We note that the description below admits some missing implementation details in favor of presenting a high-level sketch of the key ideas.

To give an intuition of how $\iterated$ works, let us first describe a constant-round protocol with a proof size of $O(\log\log \log n)$. The protocol start from a block-partition as in $\double$. An adjustment that is required here is that the blocks are of size $\Theta(\log n(\log \log n)^{2})$. Observe that verifying the validity of inner-block edges essentially reduces to applying the $\double$ protocol within each block. Due to the size of the blocks, the proof size associated with this step is $O(\log \log (\log n(\log \log n)^{2}))=O(\log \log \log n)$. 

We now focus on the outer-block edges. As described in $\double$, the task associated with outer-block edges can be reduced to a constant number of equality tasks between adjacent blocks. Furthermore, due to the block size, the prover can formulate these equality instances as instances of $(g=\Theta(\log^{2} \ell),\ell )$-equality. Given this formulation, for each block $b$, the prover and verifier can execute the self-reduction for each equality instance with which $b$ is associated, in parallel. This results in a partition into segments of size $O(\log \log n)$ such that each segment is associated with a constant number of equality instances. To test these equalities, the prover and verifier invoke the equality protocol of Lemma \ref{lemma:equality} in parallel on every equality instance of each segment. This results in a protocol with a proof size of $O(\log \log \log n)$.    

It is then straightforward to see that by slightly adjusting the block size and applying the self-reduction twice in a row, one can obtain a constant-round protocol with a proof size of $O(\log^{(4)} n)$.  More generally, by applying the self-reduction $d$ times for an arbitrary constant $d$, the result can be extended to a constant-round protocol with a proof size of $O(\log^{(d+2)} n)$. 

\smallskip
\noindent \textbf{First attempt: a na\"ive protocol.} Based on the construction described above, it might be tempting to think that we can get the guarantees of $\iterated$ (as specified in Theorem \ref{theorem:iterated-log}) simply by applying the self-reduction consecutively $\Theta(\log ^{*}n)$ times. However, this reasoning fails for two reasons. First, note that after $\Theta(\log ^{*}n)$ applications of the self-reduction, as many as $2^{\Theta(\log ^{*}n)}$ equality instances can be assigned to each segment. This is because in each invocation of the self-reduction, every equality instance is replaced with a constant number of (exponentially smaller) equality instances. Thus, after $i$ invocations, there can be as many as $2^{\Theta(i)}$ bitstrings assigned for a segment. The second reason for failure is that the reduction to equality only accounts for outer-block edges. We note that it is possible to verify the validity of the inner-block edges by recursively invoking the protocol for outer-block edges within each block. This results in $\Theta((\log ^{*}n)^{2})$ interaction rounds. Therefore, constructing the protocol in a na\"ive manner leads to $\Theta((\log ^{*}n)^{2})$ interaction rounds and a proof size of $2^{\Theta(\log ^{*}n)}$.

\smallskip
\noindent \textbf{Improving the efficiency of the na\"ive protocol.} As we show, it is possible to overcome the two obstacles described above. First, to reduce the proof size, we define a different encoding convention when the self-reduction is applied. The invariant that we maintain is that at every stage, only a constant number of bitstrings are assigned to each segment. Notice that this means in particular that after the final call to the self-reduction, we get segments of constant size such that each of them receives a constant number of bitstrings. The caveat of this encoding convention is that we don't preserve the property that equality instances are defined between pairs of adjacent segments. Indeed, following the $\Theta(\log ^{*}n)$ self-reduction applications, we may need to check equality between bitstrings assigned to non-adjacent segments. 

To face this obstacle, our goal is to implement the equality tests performed in the na\"ive protocol in an efficient manner. Towards that goal, the main idea is to exploit the structure of the $d=\Theta(\log ^{*}n)$ segment-partitions constructed during the self-reduction applications. Let us refer to the segment-partition constructed during the $i$-th call to the self-reduction as the $i$-th \emph{layer}. In the na\"ive protocol, each equality instance that originates in the $i$-th layer can be mapped to an equality instance of the $d$-th layer that corresponds to a propagation of $d-i$ evaluations of the polynomial $\Phi$ (in accordance with the parameters of the self-reduction and the segment size for each layer $i\leq j\leq d$). Therefore, the idea is to have the prover provide the nodes of an $i$-th layer segment $S$ with the propagated values of its equality instances. Given these propagated values, the nodes of $S$ can test the relevant equalities in a direct manner. Since each $i$-th layer segment has only a constant number of equality instances that originate from it, and since the bitstrings assigned in the $d$-th layer are of constant size, this requires an assignment of only constant number of bits per layer which amounts to a total communication of $O(\log ^{*}n)$ bits. Furthermore, as we show, the correctness of the assigned propagated values can be verified via a bottom-up approach. That is, for each layer $i\in [d-1]$, the nodes of an $i$-th segment can verify the validity of the assigned values based on the assignment associated with their $(i+1)$-th layer segment.

Finally, to avoid the quadratic number of interaction rounds, we show that it is possible to implement a ``piggybacking" approach. That is, we show that the recursive calls to the procedure for outer-block edges can be done \emph{in parallel} to the self-reduction invocations. To get a conceptually simple example of how this approach works, let us consider the protocol described above for proof size $O(\log \log \log n)$. In this protocol, we can actually implement $\double$ within each block in parallel to the call to the self-reduction. This only adds a constant number of equality instances for every segment and thus, does not (asymptotically) affect the proof size. Going back to the $\iterated$ protocol, applying this approach reduces the number of interaction rounds to the desired $O(\log^{*}n)$. 

One may wonder whether the consecutive invocations of the self-reduction could also be parallelized to obtain a constant-round protocol. As we explain in Appendix \ref{section:interaction}, such approach is likely to fail even when considering only two invocations of the self-reduction.     

\paragraph{On Breaking the $\log\log n$ Barrier of \cite{NaorPY20} and \cite{GilP2025}.} 
Constant-round DIPs with $O(\log \log n)$ proof size have been presented in \cite{NaorPY20} (for various tasks not related to planarity) and \cite{GilP2025} (for planarity-related tasks). Here, we note that in both works, obtaining a proof size of $o(\log \log n)$ requires a different approach and new ideas. 

We start by sketching the approach of \cite{NaorPY20}. The idea of \cite{NaorPY20} is to divide the graph into clusters of size $\Theta(\log n)$. Then, the proof involves: (1) solving the underlying task within each cluster; and (2) merging solutions between clusters to form a global solution. We note that the size of the clusters is important to obtain step (2) as it allows each cluster to act as a ``super-node'' for which the proof can be distributed among $\Theta(\log n)$ nodes. Both steps are achieved by means of a black-box invocation of a RAM compiler that is developed within that paper. On an $N$-node graph, the RAM compiler transforms any $O(N)$-time centralized algorithm into a constant-round DIP with a proof size of $O(\log N)$. Furthermore, as part of the construction, the RAM compiler incurs a $\Theta(\log N)$ additive overhead in the proof size. Since step (1) requires a call to the RAM compiler on graphs with $N=\Theta(\log n)$ nodes, it directly follows that the proof size associated with that step must be $\Theta(\log N)=\Theta(\log \log n)$. It is also worth mentioning that the approach of \cite{NaorPY20} cannot be directly applied to planarity. Indeed, as pointed out in \cite{GilP2025}, any DIP protocol for planarity that attempts to directly use the framework of \cite{NaorPY20} cannot be sound.

Regarding the approach of \cite{GilP2025}, as mentioned before, the authors reduce the planarity problem to LR-sorting. Then, the result is obtained by a protocol for LR-sorting in planar graphs with a proof size of $O(\log \log n)$. The LR-sorting  protocol of \cite{GilP2025} uses a similar block-partition to the one described above for our protocols. Then, verifying the validity of outer-block edges is done in two consecutive steps. First, the prover commits to a local proof for every outer-block edge. These local commitments are then tested in a verification scheme. As part of the commitment, the prover is required to send some value $i_{e}\in[\log n]$ for each outer-block edge $e\in E$. Therefore, the commitment part already requires a proof size of $\Theta(\log\log n)$ for every outer-block edge. Notice that in our work, we are able to solve the task of verifying the validity of outer-block edges without this commitment-verification paradigm, thus paving the way for proofs that go below the $\log\log n$ barrier.

%% file: recursive-sorting.tex
\section{LR-Sorting in Planar Graphs}\label{section:sorting}
In this section, we present new protocols for LR-sorting in planar graphs. First, the $\double$ protocol is presented in Section \ref{section:double-log}. The $\double$ protocol runs in $3$ interaction rounds and has a proof size of $O(\log \log n)$ (refer to Theorem \ref{theorem:double-log} for the full statement). Apart from its improved round-complexity compared to \cite{GilP2025}, a useful property of $\double$ is that an important sub-task of the protocol reduces to the equality task (see Section \ref{section:equality-discussion} for more details). We then focus on a variation of the equality task referred to as $(g,\ell)$-equality in Section \ref{section:self-reduction} (see the concrete definition of $(g,\ell)$-equality within the section). Particularly, Section \ref{section:self-reduction} is devoted to proving Lemma \ref{lemma:self-reduction} which describes a self-reduction result for $(g,\ell)$-equality. Finally, in Section \ref{section:recursion} we build upon the reductions of Sections \ref{section:equality-discussion} and \ref{section:self-reduction} to construct the $\iterated$ protocol which runs in $O(\log ^{*}n)$ rounds and has a proof size of $O(1)$ (refer to Theorem \ref{theorem:iterated-log} for the full statement). As we note, $\iterated$ is a special case of a family of protocols that obtain a more general trade-off between rounds and proof size (refer to Theorem \ref{theorem:general-iterated-log} for the full statement). 

We note that by combining the results stated in Theorems \ref{theorem:double-log}, \ref{theorem:iterated-log}, and \ref{theorem:general-iterated-log} with the reduction stated in Lemma \ref{lemma:sorting-to-planarity}, we prove Theorems \ref{theorem:results-three-planarity}, \ref{theorem:results-logstar-planarity}, and \ref{theorem:results-trade-off}, respectively. Moreover, by combining these results with the reduction stated in Lemma \ref{lemma:sorting-to-outerplanarity}, we prove Theorem \ref{theorem:results-outerplanar}.

 %refer to Theorem \ref{theorem:double-log} for its properties. As an additional useful property of $\double$, we show that an important sub-task of LR-sorting can be reduced to equality. Then, in Section \ref{section:self-reduction}, we present a self-reduction for (a variation of) the equality task. These reductions are then used in Section \ref{section:recursion} where the $\iterated$ protocol is presented; refer to Theorem \ref{theorem:iterated-log} for the properties of $\iterated$.  

%\ygnote{Add description; perhaps after writing results section}
\subsection{The $\double$ Protocol} \label{section:double-log}
%\ygnote{Maybe insert reduction lemma after all. Something like suppose that there is a block-partition. Verifying validity of outer-block edges reduces to equality}
In this section, we present a new protocol for LR-sorting in planar graphs to which we refer as $\double$. The properties of $\double$ are stated in the following theorem.
\begin{theorem}\label{theorem:double-log}
	The $\double$ protocol solves LR-sorting in planar graphs. It uses $3$ interaction rounds, a proof size of $O(\log \log n)$, perfect completeness, and a soundness error of $1/\log n$. 
\end{theorem}
As we elaborate below, the protocol is obtained by partitioning the edges into two sets --- namely, \emph{inner-block} edges and \emph{outer-block} edges --- and verifying their validity (w.r.t.\ LR-sorting) in parallel. As part of the $\double$ protocol, we show that the task of verifying the validity of the outer-block edges can be reduced to the equality task (as defined in Section \ref{section:arithmetic}) on instances of size $\log n$.\footnote{Throughout the paper, we may slightly abuse the notation and write, e.g., size $\log n$ instead of $\lceil \log n \rceil$.} This reduction plays an important role later on in the protocol presented in Section \ref{section:recursion}. Thus, we state its explicit details in Lemma \ref{lemma:outer-block-to-equality}. We now describe the $\double$ protocol.

%\ygnote{Change $P$ to $H$}
Consider an LR-sorting instance $(G,H)$ where $G=(V,E)$ is a directed planar graph on $n$ nodes and $H$ is a Hamiltonian path of $G$. The $\double$ protocol starts by having the prover divide the path $H$ into \emph{blocks} of size $\log n$. Each block consists of $\log n$ consecutive nodes such that the first block consists of the $\log n$ leftmost nodes in $H$, the second block consists of the next $\log n$ nodes and so on. More formally, the blocks are defined such that the $i$-th block is between the $((i-1)\log n+1)$-th node and the $(i\log n)$-th node in $H$.\footnote{Throughout the paper, whenever a path is divided into blocks we assume for simplicity that all blocks are of equal size. Lifting this assumption would result in additional case analyses but would not affect the overall validity of our results.} 

The edges of $G$ are partitioned into \emph{inner-block} edges which are the edges whose endpoints belong to the same block, and \emph{outer-block} edges which are the edges whose endpoints belong to different blocks. We note that in particular, the outer-block edges include the edges of $H$ that go from the rightmost node of a block $b$ to the leftmost node in the block $b'$ that follows $b$ in the $H$-ordering.
 
In the first round of interaction, for each edge $e\in E$, the prover indicates whether it is inner-block or outer-block by means of a single bit to $e$'s label. This is done by writing this bit to a field $\mathtt{class}(e)$ within $e$'s label. In addition, the prover informs every node whether it is the leftmost node of its block. This is done by assigning a single indicating bit as part of the label assignment produced in the first interaction. Notice that this assignment also informs each node whether it is the rightmost node of its block, i.e., whether its right path-neighbor is the leftmost node of the following block. The prover's goal is then split into two sub-tasks --- proving the validity of inner-block edges and proving the validity of outer-block edges. We show a sub-protocol for each of these tasks. The result of Theorem \ref{theorem:double-log} is then obtained by executing the sub-protocols in parallel.

\paragraph{Verifying the validity of inner-block edges.}
The sub-protocol for inner-block edges begins by having the verifier at the leftmost node of each block $b$ sample a bitstring $r_{b}$ of length $\log \log n$ bits and send it to the prover. Intuitively speaking, the role of $r_{b}$ is to act as a randomized ``identifier" for $b$. After receiving the values $r_{b}$ for each block $b$, the prover assigns each node $v\in b$ with the values $\mathtt{rand}(v)=r_{b}$ and $\mathtt{index}(v)=i_{b}(v)$ where $i_{b}(v)$ is $v$'s \emph{index} which is defined such that if $i_{b}(v)=i$, then $v$ is the $i$-th leftmost node in $b$.

For the verification, each node $u\in V$ checks that it received the same $\mathtt{rand}(\cdot)$ value as its block-neighbors and that $\mathtt{index}(u)=\mathtt{index}(u_{1})+1=\mathtt{index}(u_{2})-1$ where $u_{1},u_{2}$ are $u$'s left and right block-neighbors, respectively. In addition, the verifier at $u$ checks that $\mathtt{rand}(u)=\mathtt{rand}(v)$ and $\mathtt{index}(u)<\mathtt{index}(v)$ for every edge $e=(u,v)$ that was labeled by $\mathtt{class}(e)$ as an inner-block edge. Following the communication described above, the sub-protocol for inner-block edges ends with the verifier checking that the $\mathtt{rand}$ value received by the leftmost node of each block $b$ is identical to the sampled bitstring $r_{b}$.

%Regarding the inner-block edges, previous work \cite{GilP2025} presents a $3$-round protocol with a proof size of $O(\log \log n)$ that verifies the validity of the inner-block edges with perfect completeness and a soundness error of $1/\poly\log n$. Therefore, we now focus on the outer-block edges.

\paragraph{Verifying the validity of outer-block edges.} Towards describing the sub-protocol for outer-block edges, let us establish some additional notations and definitions. Consider the graph $G_{\text{con}}$ obtained by contracting every block into a single node and eliminating edge multiplicities and self-loops. Observe that if $G$ is planar, then so is $G_{\text{con}}$. In particular, this means that $G_{\text{con}}$ is $5$-degenerate.\footnote{Recall that a graph $G$ is $k$-degenerate if every subgraph of $G$ admits a node of degree at most $k$. We remark that for the applications in which LR-sorting is used, it is sufficient to obtain a protocol that works in outerplanar graphs in which case $G$ and $G_{\text{con}}$ are $2$-degenerate.} By a standard combinatorial argument, for every block $b$ we can define a set $D(b)$ of neighboring blocks such that: (1) $|D(b)|\leq 5$; and (2) if $b$ and $b'$ are neighboring blocks, then either $b'\in D(b)$ or $b\in D(b')$. We say that a block $b$ is \emph{accountable} for an outer-block edge $e\in E$ if $e$ has one endpoint in $b$ and another endpoint in a block $b'\in D(b)$.

For each block $b$, its \emph{position} is defined to be $i-1$ if $b$ is the $i$-th leftmost block in $H$. Let $\pos(b)$ denote the position of a block $b$ and define a partition of $D(b)$ into the sets $D^{-}(b)=\{b'\in D(b)\mid \pos(b')<\pos(b)\}$ and $D^{+}(b)=\{b'\in D(b)\mid \pos(b')>\pos(b)\}$. We can now describe the interaction rounds of the sub-protocol.

%\paragraph{Interaction round I:} The prover encodes the block partition by marking the first and last node of every block. 
\subparagraph{Interaction round I:} In the first interaction round, the prover assigns the following bitstrings to each block $b$: (1) $\pos(b)$; (2) the set of positions $S^{-}(b):=\{\pos(b')\mid b'\in D^{-}(b)\}$; and (3) the set of positions $S^{+}(b):=\{\pos(b')\mid b'\in D^{+}(b)\}$. All the assigned bitstrings are encoded in a distributed manner as described in Section \ref{section:arithmetic}. Moreover, the bitstrings in $S^{-}(b)$ (resp., $S^{+}(b)$) are encoded by the prover to the nodes of the block in some arbitrary order and we denote by $s^{-}_{i}(b)$ (resp., $s^{+}_{i}(b)$) the $i$-th bitstring of $S^{-}(b)$ (resp., $S^{+}(b)$) according to that order. Notice that since the blocks are of size $\log n$, the positions can be encoded such that each node receives a single bit for each encoded position (and a constant number of bits in total). In addition, the prover provides a proof that $\pos(b)>s^{-}_{i}(b)$ (resp., $\pos(b)<s^{+}_{i}(b))$) for every bitstring $s^{-}_{i}(b)\in S^{-}(b)$ (resp., $s^{+}_{i}(b)\in S^{+}(b)$) based on Lemma \ref{lemma:greater-than}. If any of the proofs fail, then the verifier rejects immediately.

Regarding the edge-labels, consider some outer-block edge $e$ whose endpoints belong to blocks $b$ and $b'\in D^{-}(b)$ (resp., $b'\in D^{+}(b)$). The prover assigns a label to $e$ consisting of the fields $\mathtt{acc}(e)$ and $\mathtt{ind}(e)$, where $\mathtt{acc}(e)$ points to the endpoint of $e$ that belongs to the accountable block $b$ and $\mathtt{ind}(e)$ is the index for which $s^{-}_{\mathtt{ind}(e)}(b)=\pos(b')$ (resp., $s^{+}_{\mathtt{ind}(e)}(b)=\pos(b')$). Notice that due to the edge directions, $\mathtt{acc}(e)$ can be encoded using only $1$ bit. This concludes the first interaction.

\subparagraph{Interaction rounds II and III:} For a block $b$, we denote by $\hatpos(b)$ the first bitstring assigned to $b$ in the first interaction, i.e.,  the bitstring claimed by the prover to be $\pos(b)$. We denote by $\sigma^{-}_{i}(b)$ and $\sigma^{+}_{i}(b)$ the bitstrings claimed to be $s^{-}_{i}(b)$ and $s^{+}_{i}(b)$, respectively. For every node $v\in V$, let $E^{\text{acc}}(v)$ denote the set of incident edges $e\in E(v)$ that were labeled by the prover as outer-block and the field $\mathtt{acc}(e)$ points to $v$. Let $N_{\text{in}}^{\text{acc}}(v)=\{u\mid (u,v)\in E^{\text{acc}}(v)\}$ and $N_{\text{out}}^{\text{acc}}(v)=\{u\mid (v,u)\in E^{\text{acc}}(v)\}$ be the sets of incoming and outgoing neighbors of $v$ incident on the edges of $E^{\text{acc}}(v)$. Observe that given the labeling produced by the prover in the first interaction, every node $v\in V$ can fully determine the sets $N_{\text{in}}^{\text{acc}}(v)$ and $N_{\text{out}}^{\text{acc}}(v)$.

Consider some node $v\in V$ and denote its block by $b_{v}$. The remainder of the protocol aims to check that $\sigma^{-}_{\mathtt{ind}(u,v)}(b_{v})=\hatpos(b_{u})$ for every $u\in N_{\text{in}}^{\text{acc}}(v)$, and $\sigma^{+}_{\mathtt{ind}(v,u)}(b_{v})=\hatpos(b_{u})$ for every $u\in N_{\text{out}}^{\text{acc}}(v)$, where $b_{u}$ denotes the block to which $u$ belongs. 

To test these equalities, the idea is to follow a similar approach to the one presented in the protocol of Lemma \ref{lemma:equality}. That is, to check that two bitstrings $\alpha$ and $\alpha'$ of length $\log n$ are equal, the nodes seek to check that $\Phi_{q}(r;\alpha)=\Phi_{q}(r;\alpha')$ where $q$ is a prime number in the range $[\log ^{2}n,2\log ^{2}n]$ and $r$ is drawn uniformly at random from $\{0,\dots, q-1\}$. To that end, the value $r$ is drawn and sent to the prover by the leftmost node in $H$. The prover then passes the value $r$ to all nodes in the graph. Moreover, for each block $b$ and bitstring $\alpha$ that was assigned to $b$, the prover assigns the value $\Phi_{q}(r;\alpha)$ to every node $v\in b$ along with a proof that this value is indeed the correct evaluation of $\Phi_{q}(r;\alpha)$ as described in the protocol of Lemma \ref{lemma:equality}. This assignment allows the nodes to test all equality conditions as described above. This completes the description of the protocol. We go on to analyze its complexity and correctness.

\paragraph{Complexity.} 
Notice that the sub-protocols for inner-block and outer-block edges can be executed in parallel. Since each of them operates within $3$ interaction rounds, we get that $\double$ can be executed in $3$ rounds. 

As for the proof size, in the first interaction round the prover assigns constant-sized labels. Indeed, each block receives only a constant number of bitstrings and each of them is distributed among the nodes of the block such that each node receives a single bit. Moreover, note that the edge-labels are also of constant size. Then, in the sub-protocol associated with the inner-block edges, the prover and verifier exchange $O(\log \log n)$ bits. For the sub-protocol associated with the outer-block edges, notice that each block executes the protocol of Lemma \ref{lemma:equality} once for each of its assigned bitstrings. Since each block is assigned a constant number of bitstrings each of which is of length $\log n$, we get that the proof size is $O(\log \log n)$.

\paragraph{Correctness.} The perfect completeness of $\double$ follows from the construction. We now analyze its soundness. For a no-instance of LR-sorting, we say that an edge $e=(u,v)$ is \emph{violating} if $v$ appears before $u$ in the path $H$. Recall that in the first interaction round, the prover assigns a bit $\mathtt{class}(e)$ to each edge $e\in E$ labeling it either inner-block or outer-block. We say that a label assignment is \emph{honorable} if every outer-block edge $e\in E$ is (correctly) labeled by $\mathtt{class}(e)$ as outer-block. Otherwise, it is \emph{dishonorable}. The following lemma establishes the properties of the sub-protocol for inner-block edges. The proof of the lemma is deferred to Appendix \ref{section:missing-proofs-double}.
\begin{lemma}\label{lemma:gp-soundness}
	If the label assignment in the first interaction is dishonorable or there exists a violating edge $e\in E$ labeled by $\mathtt{class}(e)$ as inner-block, then the verifier rejects the instance with probability $1-1/\log n$. 
\end{lemma}
For a label assignment given by the prover in the first interaction, we say that an edge $e=(u,v)\in E$ is \emph{lying} if $\mathtt{class}(e)$ labels it as outer-block and $u$ and $v$ belong to (not necessarily distinct) blocks $b_{u}$ and $b_{v}$ such that $\hatpos(b_{u})\geq \hatpos (b_{v})$. We now prove the following two observations.
\begin{observation}\label{observation:double-log-honorable-to-lying}
	Suppose that there exists a violating edge $e=(u,v)\in E$ that was classified by the prover as outer-block. Then, any honorable label assignment admits a lying edge.
\end{observation}
\begin{proof}
	First, consider the case that $e$ is an inner-block edge (and was wrongfully labeled by $\mathtt{class}(e)$). In this case, the observation holds trivially since $e$ is a lying edge for any label assignment. Now, suppose that $e$ is an outer-block edge. Since $e$ is a violating edge, adding it to the subpath of $H$ from $v$ to $u$ closes a cycle $C$. Let us denote by $b_{0},\dots , b_{k-1}$ the blocks that participate in $C$ such that for all $0\leq i\leq k-1$, there exists a $C$-edge from a node in $b_{i}$ to a node in $b_{(i+1)\bmod k}$. Due to the cyclicity, for any assignment of $\hatpos$ values to the blocks, there exists an index $0\leq i\leq k-1$ for which $\hatpos(b_{i})\geq \hatpos(b_{(i+1)\bmod k})$. For such $i$, the $C$-edge from a node in $b_{i}$ to a node in $b_{(i+1)\bmod k}$ is by definition a lying edge.
\end{proof}
\begin{observation}\label{observation:double-log-lying-to-rejection}
	Suppose that the label assignment given by the prover in the first interaction admits a lying edge. Then, the verifier rejects the instance with probability $1-1/\log n$.
\end{observation}
\begin{proof}
	Suppose that $e=(u,v)$ is a lying edge and let $b_{u}$ and $b_{v}$ be the blocks to which $u$ and $v$ belong, respectively. Let us assume that $\mathtt{acc}(e)$ points to $u$ (the case where $\mathtt{acc}(e)$ points to $v$ is analogous) and let $i$ denote the index given in $\mathtt{ind}(e)$. We can assume that $\sigma_{i}^{+}(b_{u})>\hatpos(b_{u})$ (as this is checked by the verifier in a direct manner following interaction round I). Since $e$ is a lying edge, we get that $\sigma_{i}^{+}(b_{u})>\hatpos(b_{u})\geq \hatpos(b_{v})$ and in particular, $\sigma_{i}^{+}(b_{u})\neq  \hatpos(b_{v})$. By construction, the verifier rejects unless $\Phi_{q}(r;\sigma_{i}^{+}(b_{u}))= \Phi_{q}(r;\hatpos(b_{v}))$ which by Lemma \ref{lemma:poly-identity}, occurs with probability at most $(\log n)/q<1/\log n$.
	%construction, the verifier rejects unless $\Phi_{p}(r;\sigma_{i}^{+}(b_{u}))=\Phi_{p}(r;\hatpos(b_{v}))$. By Lemma \ref{lemma:poly-identity}, this event occurs with probability at most $\log n/p<1/\log ^{9}n$.
\end{proof}
It is now simple to establish the protocol's soundness.
\begin{lemma}[Soundness]\label{double-log-soundness}
	If $(G,H)$ is a no-instance, then the verifier of $\double$ rejects the instance with probability $1-1/\log n$ against any prover.
\end{lemma}
\begin{proof}
	Suppose that $(G,H)$ is a no-instance, i.e., it admits a violating edge $e\in E$. If $e$ is labeled by $\mathtt{class}(e)$ as inner-block or the prover produces a dishonorable labeling, then by Lemma \ref{lemma:gp-soundness} we get that the verifier rejects the instance with probability $1-1/\log n$. So, suppose that a violating edge $e$ is labeled by $\mathtt{class}(e)$ as outer-block and that the labeling provided by the prover is honorable. Then, it follows from Observations \ref{observation:double-log-honorable-to-lying} and \ref{observation:double-log-lying-to-rejection} that the verifier rejects the instance with probability $1-1/\log n$ which completes the soundness proof.
\end{proof}
%\ygnote{Not sure about Lemma \ref{lemma:outer-block-to-equality}. Revisit that.} 
%\paragraph{On the reduction to equality.}
\subsubsection{On the Reduction to Equality}\label{section:equality-discussion}
 As described above, the sub-protocol associated with verifying the validity of outer-block edges is facilitated by a reduction to the equality task. This reduction is defined by the interaction of the first round of $\double$. Since this reduction plays an important role in the $\iterated$ protocol of Section \ref{section:recursion}, we explicitly state its properties in the following lemma.
\begin{lemma}[Derived from the analysis of $\double$]\label{lemma:outer-block-to-equality}
	Let $\mathsf{EQ}_{k}$ be a protocol with perfect completeness for equality on instances of size $k$. Let $(G,H)$ be an instance of LR-sorting which consists of a directed planar graph $G=(V,E)$ on $n$ nodes and a Hamiltonian path $H$ in $G$. Moreover, suppose that $H$ is divided into blocks of size $\ell\geq \log n$ and let $F$ be the set of outer-block edges. Then, there exists a protocol $\Pi$ for the validity of $F$ (w.r.t.\ LR-sorting) such that:
	\begin{itemize}
		\item $\rounds(\Pi)=\rounds(\mathsf{EQ}_{\ell})+1-(\rounds(\mathsf{EQ}_{\ell})\bmod 2)$
		\item $\size(\Pi)=O(\size(\mathsf{EQ}_{\ell}))$
		\item $\completeness(\Pi)=\completeness(\mathsf{EQ}_{\ell})=0$
		\item $\soundness(\Pi)=\soundness(\mathsf{EQ}_{\ell})$
	\end{itemize}
\end{lemma}
Notice that $\rounds(\Pi)$ depends on the parity of $\rounds(\mathsf{EQ}_{\ell})$. This is because if $\mathsf{EQ}_{\ell}$ starts with a verifier-interaction (i.e., $\rounds(\mathsf{EQ}_{\ell})$ is even), then $\Pi$ needs an additional round for prover-interaction (for instance, this is the case in the $\double$ protocol). In contrast, if  $\mathsf{EQ}_{\ell}$ starts with a prover-interaction, then the first round of $\mathsf{EQ}_{\ell}$ can be merged with the interaction that produces the equality instances.

\subsection{The Self-Reducibility of Equality (with Gaps)}\label{section:self-reduction}
%%%%%%%%%%%%%%%%%%%%%%%%%%%%%%%%%%%%%%%%%%%%%%%%%%%%%%%%%%%%%%%%%%%%%%%%%%%%%%
%\ygnote{Maybe change blocks to clusters/components in this section.}
In this section, we present a self-reduction for a variation of equality to which we refer as \emph{$(g,\ell)$-equality}. This variation differs from the standard equality task strictly by the manner in which a bitstring $\beta$ is encoded on a path $Q$. Specifically, the task is defined such that $\beta$ is of length $\ell$ and $Q$ is of size $g\cdot \ell$ for two integers $g,\ell>0$. The encoding is then defined such that $\beta$ is written on $Q$ such that the bit $\beta[i]$ is written on the $((i-1)\cdot g+1)$-th leftmost node of $Q$. In other words, the bitstrings are encoded such that every pair of consecutive bits are separated by $g-1$ nodes. 

Based on this encoding paradigm, the $(g,\ell)$-equality task is defined similarly to the standard equality tasks. That is, the goal is to test the equality between two bitstrings $\alpha,\alpha'\in \{0,1\}^{\ell}$ which are encoded to an ordered pair $(P,P')$ of node-disjoint paths of size $g\cdot\ell$ according to the mechanism described above.
\begin{lemma}\label{lemma:self-reduction}
	There exist two constants $c_{1},c_{2}>0$ for which the following holds. Let $g,\ell>0$ be two integers and suppose that there exists a $(g,c_{1} \log \ell)$-equality protocol $\Pi$ with $\completeness(\Pi)=0$. Then, there exists a $(c_{2}\cdot g\cdot \log ^{2} \ell,\ell)$-equality protocol $\Pi'$ such that: 
	\begin{itemize}
		\item $\rounds(\Pi')=\rounds(\Pi)+3-(\rounds(\Pi)\bmod 2)$
		\item $\size(\Pi')=O(\size(\Pi))$
		\item $\completeness(\Pi')=0$
		\item $\soundness(\Pi')\leq \soundness(\Pi)+1/\ell$
	\end{itemize}
\end{lemma}
For the sake of clear presentation, we keep the constants $c_{1}$ and $c_{2}$ implicit in the protocol's description. The proof of Lemma \ref{lemma:self-reduction} is presented in two stages. First, in Section \ref{section:simplified} we present a protocol that operates under some simplifying assumptions. We emphasize that the simplifying assumptions are placed strictly to avoid a cumbersome description. In Appendix \ref{section:lifting}, we present the modifications required in order to lift the simplifying assumptions.

We note that the guarantees of the protocol presented in Section \ref{section:simplified} will be slightly better than those stated in Lemma \ref{lemma:self-reduction}. In particular, the number of rounds in the described protocol will be $\rounds(\Pi)+1-(\rounds(\Pi)\bmod 2)$ and the soundness error will be $\soundness(\Pi)+1/2\ell$. The number of interaction rounds and soundness error that are stated in Lemma \ref{lemma:self-reduction} will become necessary when we lift the simplifying assumptions in Appendix \ref{section:lifting}.

%\ygnote{maybe index blocks to start from $0$}
%%%%%%%%%%%%%%%%%%%%%%%%%%%%%%%%%%%%%%%%%%%%%%%%%%%%%%%%%%%%%%%%%%%%%%%%%%%%%%
\subsubsection{A Protocol with Simplifying Assumptions}\label{section:simplified}
%%%%%%%%%%%%%%%%%%%%%%%%%%%%%%%%%%%%%%%%%%%%%%%%%%%%%%%%%%%%%%%%%%%%%%%%%%%%%%
The first simplifying assumption we make in the current section is that $g=1$. That is, we consider an instance of $(c_{2}\log ^{2} \ell,\ell)$-equality on two node-disjoint paths $P,P'$ of size $c_{2}\cdot \log ^{2} \ell\cdot \ell$. In addition, we assume that two integers  $q$ and $r$ are encoded to the $c_{1}\log \ell$ leftmost nodes of $P$ and $P'$ in a distributed manner, where $q$ is a prime number in the range $[2\ell^{2},3\ell^{2}]$ and $r$ is a random value drawn u.a.r.\ from $\{0,\dots, q-1\}$.

%\ygnote{Maybe this paragraph is for the overview section.}
Let $\alpha$ and $\alpha'$ denote the $\ell$-length bitstrings encoded to $P$ and $P'$, respectively, in accordance with the formulation of $(c_{2}\log ^{2} \ell,\ell)$-equality. The idea of the protocol is simple. The prover seeks to assign the nodes of $P$ and $P'$ with the values $\Phi_{q}(r;\alpha)$ and $\Phi_{q}(r;\alpha')$ in a distributed manner and then use $\Pi$ to prove that the bitstrings representing $\Phi_{q}(r;\alpha)$ and $\Phi_{q}(r;\alpha')$ are equal. The main objective of the protocol is therefore to have the prover provide proof that the polynomials were computed correctly. As we will show, this task can be reduced to equality as well.

The protocol starts by having the prover divide both $P$ and $P'$ into \emph{clusters} of size $c_{2}\log ^{2} \ell$. The clusters consist of consecutive path nodes such that the $i$-th cluster starts from the $(c_{2}\log ^{2} \ell\cdot(i-1)+1)$-th node of the path and finishes at the $(c_{2}\log ^{2} \ell\cdot i)$-th node of the path. Notice that the clusters are defined such that for each $1\leq i\leq \ell$, the bits $\alpha[i]$ and $\alpha'[i]$ are written on the first node of the $i$-th cluster in $P$ and $P'$, respectively. Each cluster is further partitioned into \emph{segments} of size $c_{1}\log \ell$. Similarly to the cluster-partition, each segment consists of $c_{1}\log \ell$ consecutive cluster-nodes. As part of the label assignment (which is described in full below), the prover informs each node whether it is the leftmost node of a segment/cluster. This only requires a constant number of bits. We now describe the prover's interaction as well as the verifier's verification process on the path $P$ (the description for $P'$ is similar).

For each $1<i\leq\ell$, the prover assigns the following bitstrings to every segment of the $i$-th cluster of $P$: $w_{1}(i)=q$;  $w_{2}(i)=r$;  $w_{3}(i)=r^{i-1}\bmod q$;  $w_{4}(i)=r^{i}\bmod q$;  $w_{5}(i)=\Phi_{q}(r;\alpha [1,\dots, i-1])$;  $w_{6}(i)=\Phi_{q}(r;\alpha [1,\dots, i])$; and $w_{7}(i)=\Phi_{q}(r;\alpha)$. The values $w_{j}(i)$ are assigned to every segment of cluster $i$ in a distributed manner such that every node receives a bit for each $w_{j}(i)$ (and $7$ bits in total). For the first cluster, the prover only assigns the values $w_{1}(1)=q,w_{2}(1)=w_{4}(1)=r,w_{6}(1)=\alpha[1]\cdot r,w_{7}(1)=\Phi_{q}(r;\alpha)$. Notice that due to the simplifying assumptions, the verifier at the first cluster can directly check the correctness of all assignment apart from $w_{7}(1)$.

Given the label assignment, the verifier at each cluster $1\leq i\leq \ell$ seeks to verify \emph{segment-consistency}, i.e., that all segments of cluster $i$ received the exact same values $w_{1}(i),\dots, w_{7}(i)$. In addition, the verifier at each cluster $1< i\leq \ell$ aims to check \emph{cluster-consistency}, i.e., that all of the following equalities are satisfied: (1) $w_{1}(i)=w_{1}(i-1)$; (2) $w_{2}(i)=w_{2}(i-1)$; (3) $w_{7}(i)=w_{7}(i-1)$; (4) $w_{3}(i)=w_{4}(i-1)$; (5) $w_{5}(i)=w_{6}(i-1)$; (6) $w_{6}(i)=w_{5}(i)+\alpha[i]\cdot w_{4}(i)\bmod w_{1}(i)$ and (7) $w_{4}(i)=[w_{3}(i)\cdot w_{2}(i)]\bmod w_{1}(i)$. Notice that conditions (4) and (5) of cluster-consistency are required since $w_{3}(i)$ and $w_{4}(i-1)$ are both claimed to be $r^{i-1}\bmod q$, and $w_{5}(i)$ and $w_{6}(i-1)$ are both claimed to be $\Phi_{q}(r;\alpha [1,\dots, i-1])$. In addition to the above conditions, the verifier at cluster $\ell$ checks that $w_{6}(\ell)=w_{7}(\ell)$ in a direct manner via a bit-by-bit comparison.

The segment-consistency conditions are verified within each cluster $1\leq i\leq \ell$ by an application of the $(1,c_{1}\log \ell)$-equality protocol $\Pi$ between every pair of adjacent segments of cluster $i$ for every assigned bitstring $w_{j}(i)$. Similarly, the cluster-consistency conditions (1)--(5) are verified for each $1<i\leq \ell$ by applications of $\Pi$ between the rightmost segment of cluster $i-1$ and the leftmost segment of cluster $i$.

Verifying condition (6) for each cluster $i$ is simple. If $\alpha[i]=0$, then the verifier can directly check that $w_{5}(i)=w_{6}(i)$. Otherwise, the verifier uses the modular addition scheme of Lemma \ref{lemma:addition-mod-p} to verify that $w_{6}(i)=[w_{5}(i)+ w_{4}(i)]\bmod w_{1}(i)$. As for (7), we note that this is a task of verifying modular multiplication. Since the clusters are of size $c_{2}\log ^{2}\ell$, Lemma \ref{lemma:multiplication-to-equality} implies that multiplication can be reduced to a constant number of $(1,c_{1}\log \ell)$-equality instances between all pairs of adjacent segments in cluster $i$. Furthermore, Lemma  \ref{lemma:multiplication-mod-p} extends this to modular multiplication. Once the equality instances are constructed, the protocol $\Pi$ is used to verify their correctness.

Towards the final verification step, recall that by definition of the equality task, there exists an edge $e^{*}=(u^{*},v^{*})$ between $u^{*}\in P$ and $v^{*}\in P'$. Let $w_{7}(u^{*})$ and $w_{7}(v^{*})$ be the $w_{7}(\cdot)$ values assigned to the segments of $u^{*}$ and $v^{*}$, respectively. Recall that these values are claimed by the prover to be $\Phi_{q}(r;\alpha)$ and $\Phi_{q}(r;\alpha')$, respectively. To complete the protocol, the verifier checks that $w_{7}(u^{*})=w_{7}(v^{*})$ by means of an execution of $\Pi$ between the segments of $u^{*}$ and $v^{*}$. This completes the description of the protocol $\Pi'$. We now analyze its properties.

\paragraph{Complexity.} 
The protocol $\Pi'$ is obtained by a single round of interaction from prover to verifier, followed by an invocation of $\Pi$ on many instances in parallel. Notice that if $\Pi$ begins with a prover interaction (i.e., $\rounds(\Pi)$ is odd), then the prover can merge the first interactions of $\Pi$ and $\Pi'$. This implies that if $\rounds(\Pi)$ is odd, then $\rounds(\Pi')=\rounds (\Pi)$. If $\rounds(\Pi)$ is even, then $\Pi$ begins immediately after the initial prover interaction which means that in this case $\rounds(\Pi')=\rounds (\Pi)+1$. For the proof size, the assignment of $w_{j}(i)$ values to every segment of cluster $i$ assigns each node with a constant number of bits. Then, the protocol is designed such that each segment participates in a constant number of $\Pi$ executions (some of which come from the verification of modular multiplication). Thus, we get $\size(\Pi')=O(\size (\Pi))$.

\paragraph{Correctness.}  
For the completeness of the protocol, suppose that $\alpha=\alpha'$. This means that $\Phi_{q}(r;\alpha)=\Phi_{q}(r;\alpha')$ and it follows that the label assignment produced by the honest prover satisfies all verification conditions. Since the verification conditions are instances of equality whose correctness is verified using $\Pi$, we get that $\completeness(\Pi')=\completeness(\Pi)=0$. 

For soundness, suppose that $\alpha\neq\alpha'$. This means that the probability of the event $\Phi_{q}(r,\alpha)\neq\Phi_{q}(r,\alpha')$ is at least $1-\ell/q>1-1/2\ell$. Let us now condition on this event and bound the probability that the verifier accepts. Recall that in the final verification step $\Pi$ is used to check whether $w_{7}(u^{*})=w_{7}(v^{*})$, where $w_{7}(u^{*})$ is the value assigned to $u^{*}$'s segment under the claim $w_{7}(u^{*})=\Phi_{q}(r,\alpha)$ and $w_{7}(v^{*})$ is the value assigned to $v^{*}$'s segment under the claim that $w'_{7}(v^{*})=\Phi_{q}(r,\alpha')$. Therefore, if the labels assigned by the prover satisfy $w_{7}(u^{*})\neq w_{7}(v^{*})$, then the probability that the verifier accepts the instance is bounded by $\soundness(\Pi)$. 

Let us now assume that $w_{7}(u^{*})= w_{7}(v^{*})$ and note that since we are considering the case where $\Phi_{q}(r,\alpha)\neq\Phi_{q}(r,\alpha')$, this implies that either $w_{7}(u^{*})\neq \Phi_{q}(r,\alpha)$ or $w_{7}(v^{*})\neq \Phi_{q}(r,\alpha')$. W.l.o.g., we assume the former. Now, notice that if all segment-consistency and cluster-consistency conditions are satisfied, then it would imply that $w_{7}(u^{*})=\Phi_{q}(r,\alpha)$. Therefore, the case where $w_{7}(u^{*})\neq \Phi_{q}(r,\alpha)$ implies that there is a verification condition that does not hold in $P$. Since we use $\Pi$ to verify all equality conditions, we get that the probability of the verifier accepting the instance is bounded by $\soundness(\Pi)$. 

To conclude, when conditioning on the event $\Phi_{q}(r,\alpha)\neq\Phi_{q}(r,\alpha')$ the acceptance probability is at most $\soundness(\Pi)$. Since the probability of $\Phi_{q}(r,\alpha)=\Phi_{q}(r,\alpha')$ is bounded by $1/2\ell$, it follows from a union bound argument that $\soundness(\Pi')\leq \soundness(\Pi)+1/2\ell$.

%\ygnote{Maybe write here about the convention in the next section where the prover assigns $q$ and $r$ only to the leftmost nodes.}

\subsection{The $\iterated$ Protocol}\label{section:recursion}
In this section, we present the $\iterated$ protocol. The properties of the protocol are stated in the following theorem.
\begin{theorem}\label{theorem:iterated-log}
	The $\iterated$ protocol solves LR-sorting in planar graphs. It admits $O(\log ^{*}n)$ interaction rounds, a proof size of $O(1)$, perfect completeness, and a soundness error of $1/8$.
\end{theorem}
As we describe below, the $\iterated$ protocol is obtained from $d=\Theta(\log ^{*}n)$ invocations of the procedure described in Lemma \ref{section:self-reduction}. For the sake of a clear exposition, we present the protocol for a specific $d$ value. However, we note that changing the value of $d$, along with some minuscule adjustments, one can obtain the following generalized result.  
\begin{theorem}\label{theorem:general-iterated-log}
	For any integer $1\leq d\leq \log^{*}n$, there exists a protocol $\Pi_{d}$ that solves LR-sorting in planar graphs such that $\rounds(\Pi_{d})=O(d)$; $\size(\Pi_{d})=O(\log ^{(d+1)}n)$; $\completeness(\Pi_{d})=0$; and $\soundness(\Pi_{d})=O(1/\log ^{(d)}n)$.
\end{theorem}
We present the $\iterated$ protocol in two stages. First, a na\"ive protocol that runs for $O((\log ^{*}n)^{2})$ rounds and admits a proof size of $2^{O(\log ^{*}n)}$ is presented in Section \ref{section:iterated-naive}. Then, in Section \ref{section:iterated-better}, we show how the na\"ive protocol can be modified to obtain the properties stated in Theorem \ref{theorem:iterated-log}. 

\subsubsection{A Protocol with $O((\log ^{*}n)^{2})$ Rounds and $2^{O(\log ^{*}n)}$ Proof Size}\label{section:iterated-naive}
Let $(G,H)$ be an instance for the LR-sorting problem which consists of a directed planar graph $G$ on $n$ nodes along with a Hamiltonian path $H$ in $G$. Let $d$ denote the smallest integer for which $\log^{(d)}n\leq 16$ and notice that $d=\Theta(\log ^{*}n)$. For every integer $j\in \{1,\dots , d\}$, define $g_{j}=\prod_{i=j+1}^{d+1}c (\log^{(i)}n)^{2}=c^{d-j}\prod_{i=j+1}^{d+1} (\log^{(i)}n)^{2}$, where $c$ is a constant which is large enough to accommodate the protocol's construction. 

In the beginning of the protocol, the prover divides $H$ into blocks of size $g_{1}\log n$. As described in the $\double$ protocol of Section \ref{section:double-log}, the block construction induces a partition of the edges into inner-block and outer-block. As in $\double$, let us assume that each edge $e\in E$ is assigned with a single bit $\mathtt{class}(e)$ that indicates if it is inner-block or outer-block. We will now focus on the validity of the outer-block edges and later show how to verify the validity of the inner-block edges.

\paragraph{Verifying the validity of outer-block edges.} In the first interaction round, the prover assigns each block with a constant number of equality instances based on the reduction described in Lemma \ref{lemma:outer-block-to-equality} (recall that the explicit details of the reduction are presented in the first interaction round of the $\double$ protocol). A modification that we make is that each of the bitstrings assigned to a block $b$ is encoded as a $(g_{1},\log n)$-equality instance (as defined in Section \ref{section:self-reduction}). It is straightforward to see that Lemma \ref{lemma:outer-block-to-equality} applies for this case as well. 

A small subtlety that we need to address is that as part of the reduction, the prover needs to provide a proof for claims of the form $\alpha>\alpha'$ associated with pairs of bitstrings $\alpha,\alpha'$ assigned to block $b$. This is accomplished in $\double$ with $O(1)$ proof size based on the protocol of Lemma \ref{lemma:greater-than}. Here, we cannot use Lemma \ref{lemma:greater-than} as-is because it assumes that the node associated with the bit $\alpha[i]$ is adjacent to the node associated with the bit $\alpha[i+1]$ for each index $i$. This is no longer the case in instances of $(g_{1},\log n)$-equality. Nevertheless, it is easy to simulate the scheme of Lemma \ref{lemma:greater-than} by having the prover assign the labels associated with $\alpha[i]$ and $\alpha[i+1]$ to all $g_{1}-1$ nodes that are placed between $\alpha[i]$ and $\alpha[i+1]$ in the encoding. This allows each block to simulate the protocol of Lemma \ref{lemma:greater-than} as if the bits $\alpha[i]$ and $\alpha[i+1]$ are assigned to adjacent nodes.

Notice that the $g_{j}$ values are define such that we can use Lemma \ref{lemma:self-reduction} to reduce each $(g_{1},\log n)$-equality to a constant number of $(g_{2},O(\log \log n))$-equality tasks which are defined between pairs of neighboring segments of size $O(\log \log n)$. More generally, we can recursively define the execution of the protocol in \emph{layers} $1,\dots, d$ as follows. The first layer is defined by the block construction and $(g_{1},\log n)$-equality instances described above. Then, for each $1\leq j\leq d-1$, the $(j+1)$-th layer is obtained from the $j$-th layer by an invocation of Lemma \ref{lemma:self-reduction} to reduce each $(g_{j},O(\log ^{(j)}n))$-equality instance of layer $j$ to a constant number of $(g_{j+1},O(\log ^{(j+1)} n))$-equality tasks which are defined between pairs of neighboring segments of size $O(\log ^{(j+1)} n)$. 

Notice that by construction, in the $j$-th layer, segments are of size $O(\log ^{(j)}n)$ and each segment is affiliated with $2^{O(j)}$ bitstrings which are associated with instances of $(g_{j},O(\log ^{(j)}n))$-equality. In particular, segments of the $d$-th layer are of constant size and each segment is encoded with $2^{O(d)}=2^{O(\log ^{*}n)}$ bitstrings which are associated with instances of $(g_{d},O(\log ^{(d)}n))$-equality. Recall that by definition $g_{d}=O(1)$ and $\log ^{(d)}n=O(1)$. Hence, the remaining task of the sub-protocol is to check the correctness of $2^{O(\log ^{*}n)}$ instances of $(O(1),O(1))$-equality for every $d$-th layer segment.

For each $d$-th layer segment $\mathtt{seg}$, let $A(\mathtt{seg})$ be the sequence of bitstrings assigned to it which are associated with equality instances. In the final interaction round, the prover assigns each node $v\in \mathtt{seg}$ with the entire sequence $A(\mathtt{seg})$ along with its position within $\mathtt{seg}$. We note that this assignment allows each node $v\in \mathtt{seg}$ to verify that it received a correct $A(\mathtt{seg})$ assignment. This is done by having $v$ check that it received the same  $A(\mathtt{seg})$ assignment as its segment-neighbors and that its assigned position is consistent with its segment-neighbors. In addition, $v$ checks that the bits assigned to it in the distributed encoding of $A(\mathtt{seg})$ are consistent with its received $A(\mathtt{seg})$ assignment. Following that, the nodes of each segment $\mathtt{seg}$ perform the equality tasks defined in the procedure of Lemma \ref{section:self-reduction} directly by comparing the relevant bitstrings in $A(\mathtt{seg})$ with nodes of neighboring segments. This concludes the sub-protocol for verifying the validity of outer-block edges. We now describe the sub-protocol for inner-block edges.

\paragraph{Verifying the validity of inner-block edges.} The sub-protocol for inner-block edges starts with the leftmost node of each block $b$ sampling a random bitstring $r_{b}$ of length $3$ and sending it to the prover. The prover then assigns each node $v\in b$ with the value $\mathtt{rand}(v)=r_{b}$. Each node $v$ then checks that it received the same $\mathtt{rand}$ value as its block-neighbors and that $\mathtt{rand}(v)=\mathtt{rand}(u)$ for every edge $(u,v)\in E$ that was labeled as inner-block by the $\mathtt{class}(e)$ assignment. 

The sub-protocol is then obtained by means of a recursive invocation of the protocol described so far on all subgraphs $G(b)$ induced by a block $b$. The recursive calls on the subgraphs are executed in parallel. The recursion halts when the size of the evaluated graph becomes $2$ in which case LR-sorting trivially holds. This completes the description of the protocol. We go on to analyze it.

\paragraph{Complexity.} 
The sub-protocol for outer-block edges consists of $d=O(\log ^{*}n)$ layers each of which is obtained by a call to the procedure of Lemma \ref{lemma:self-reduction} which takes $O(1)$ rounds. Therefore, the total number of rounds in the sub-protocol is $O(\log ^{*}n)$. As for the sub-protocol for inner-block edges, first notice that for an $n$-node graph, the block size is defined to be $g_{1}\log n$  where $g_{1}=\prod_{i=j+1}^{d}c (\log^{(i)}n)^{2}<c^d\cdot (\log \log n)^{2d}\leq(c\log \log n)^{2\log ^{*}n}<\log n$. Therefore, the block size can be bounded by $O(\log^{2}n)$ and therefore, it follows that the recursion depth is $O(\log ^{*}n)$. Overall, we get that the protocol runs for $O((\log ^{*}n)^{2})$ rounds. As for proof size, the largest labels assigned during the protocol are the sets $A(\mathtt{seg})$ sent to each segment $\mathtt{seg}$ in the $d$-th layer of the sub-protocol for outer-block edges. Therefore, the proof size is $O(|A(\mathtt{seg})|)=2^{O(d)}=2^{O(\log ^{*}n)}$.

\paragraph{Correctness.}
By the completeness of the $\double$ protocol, if $(G,H)$ is a yes-instance, then all equality instances defined in the first layer for outer-block edges are correct for all invocations of the sub-protocol for outer-block edges. Then, by the completeness in Lemma \ref{lemma:self-reduction}, it follows that all equality instances evaluated in the $d$-th layer are correct. Thus, it follows that the verifier accepts the instance in this case, i.e., the protocol has perfect completeness.

Regarding soundness, first suppose that the prover produces a dishonorable assignment (as defined in the correctness analysis of Section \ref{section:double-log}) for some block-partition defined during the protocol (i.e., either the initial block-partition or one invoked during a recursive step). This means that there is an outer-block edge $(u,v)$ that is labeled as inner-block. Let $b$ and $b'$ denote the block of $u$ and $v$ respectively. Then, the verifier rejects unless $r_{b}=r_{b'}$ which occurs with probability $1/8$. Hence, we now assume that throughout the protocol whenever the prover constructs a block-partition, it provides an honorable assignment. 

Suppose that $e$ is a violating edge and notice that $e$ must be an outer-block edge for some block-partition defined during the protocol. By the assumption, this means that $e$ is (correctly) labeled as outer-block at that stage. Let $G'$ be the graph for which $e$ is an outer-block edge and let $N$ be the number of nodes in $G'$. By Observation \ref{observation:double-log-honorable-to-lying}, the (honorable) label assignment admits a lying edge. As noted in the proof of Observation \ref{observation:double-log-lying-to-rejection}, this implies that there exists an incorrect equality instance in layer $1$ of the execution of the sub-protocol for outer-block edges on $G'$. Furthermore, by Lemma \ref{lemma:self-reduction} we get that if there exists an incorrect equality instance in layer $j$, then the probability that all equality instances are correct in layer $j+1$ is at most $1/\log ^{(j)}N$. Let $d'$ denote the final layer of the execution and recall that $d'$ is the smallest integer for which $\log ^{(d')}N\leq 16$. By a union bound argument, we get that the probability that all equality instances are correct in the $d'$-th layer is bounded from above by $1/\log N+1/\log ^{(2)}N+\dots+1/\log ^{(d'-1)}N$. Now, notice that for every integer $1\leq j\leq d'-1$, it holds that $1/\log^{(j)}N\leq (1/16)\cdot 1/2^{d'-j-1}$. Therefore, we can bound the probability that all equality instances of the $d$-th layer are correct by $(1/16)\cdot \sum_{j=1}^{d'-1}1/2^{d'-j-1}=(1/16)\cdot\sum_{i=0}^{d'-2}1/2^{i}<(1/16)\cdot \sum_{i=0}^{\infty}1/2^{i}=1/8$. This concludes the soundness argument since by construction, the verifier rejects the instance unless all equality instances of the $d$-th layer are correct. 

\subsubsection{A Protocol with $O(\log ^{*}n)$ Rounds and $O(1)$ Proof Size}\label{section:iterated-better}
%We present some modifications to the sub-protocol for outer-block edges from Section \ref{section:iterated-naive}. As we will show, this modification reduces the proof size of the protocol to $O((\log ^{*}n)^{2})$ while (asymptotically) preserving the number of interaction rounds. 
We present some modifications to the protocol of Section \ref{section:iterated-naive}. As we show, applying these modifications produces a protocol that satisfies the properties stated in Theorem \ref{theorem:iterated-log}. For convenience, we describe a protocol where the final interaction round admits labels of size $\Theta(d)$ and every other round admits constant sized labels. To get to a constant proof size, the prover can break the labels of the final round into $d$ constant sized pieces, and assign these pieces to the nodes over $d$ interaction rounds.

\paragraph{Reducing the proof size.}
Recall that in the sub-protocol for outer-block edges in Section \ref{section:iterated-naive}, we construct $d$ layers in $O(d)$ interaction rounds. For each $1\leq j\leq d-1$, the protocol advances from layer $j$ to layer $j+1$ by an invocation of Lemma \ref{lemma:self-reduction}. This transition exponentially reduces the layer's segment but increases the number of bitstrings encoded to each segment by a constant factor. In this section, we avoid this increase in encoded bitstrings by defining a different encoding convention when advancing from layer $j$ to layer $j+1$. 

Suppose that in layer $j\in [d]$, a segment $S$ receives $k$ bitstrings $\alpha_{0},\dots, \alpha_{k-1}$ as part of the layer encoding. For each $i\in \{0,\dots, k-1\}$, let $\mathcal{T}_{i}$ be the segment-partition obtained from invoking the procedure of Lemma \ref{lemma:self-reduction} on the segment $S$ w.r.t.\ the bitstring $\alpha_{i}$. The segment-sequence of layer $j+1$ that is placed within $S$ is defined as an interleaving of the sequences $\mathcal{T}_{0},\dots, \mathcal{T}_{k-1}$ such that for each $i\in  \{0,\dots, k-1\}$, a segment $T\in\mathcal{T}_{i}$ is immediately followed by a segment $T'\in\mathcal{T}_{i+1\bmod k}$. In other words, instead of a single segment which receives $k\cdot \delta$ bitstrings for some constant $\delta>1$, we have a \emph{batch} $B\in \mathcal{T}_{0}\times \dots\times\mathcal{T}_{k-1}$ of $k$ segments each of which receives $\delta$ bitstrings. We remark that this encoding introduces the following obstacle: it is no longer the case that in every layer $j$ all equalities are defined between neighboring segments. This is because consecutive segments in $\mathcal{T}_{i}$ are no longer adjacent in the graph (as they are separated by other segments in their respective batches). 

%As we will show, we will be able to deal with this obstacle due to some structural properties of the layers.
Let us make some additional remarks regarding the described encoding. First, we note that by construction, this encoding maintains the invariant that for all layers of the sub-protocol, every segment receives a constant number of bitstrings. Consequently, the size of each batch is constant (since it corresponds to the number of bitstrings assigned to a single segment in the previous layer). We also note that each $j$-th layer segment is associated with a constant number of invocations of Lemma \ref*{lemma:self-reduction}. Since each such invocation is mapped to a distinct segment-sequence of layer $j+1$, we need to increase the value $g_{j}$ by a multiplicative factor $2^{O(d-j)}$ for each $1\leq j\leq d$. Observe that this adjustment can be obtained simply by increasing the constant $c$ in the definition of $g_{j}$ in Section \ref{section:iterated-naive}. 

Towards describing the remainder of the protocol, we establish some additional notations. Recall that when applying the procedure of Lemma \ref{lemma:self-reduction} w.r.t.\ some bitstring $\alpha$, the prover assigns a field $w_{7}(\cdot)$ which corresponds to $\Phi_{q}(r;\alpha)$ to each segment. For each bitstring $\alpha$ assigned to some segment in layer $j\in [d]$, let $D(\alpha)$ denote the bitstring assigned by the prover to a segment in layer $d$ which was obtained by a sequence of $w_{7}(\cdot)$ assignments. Notice that $D(\alpha)=\alpha$ for every bitstring $\alpha$ that was assigned in the $d$-th layer.

Consider a segment $S$ of layer $j\in [d]$. We denote by $A(S)$ the bitstring sequence assigned to $S$ and by $D(S)$ the sequence of $D(\alpha)$ values for bitstrings $\alpha\in A(S)$. For a batch $B=(S_{1},\dots, S_{|B|})$ consisting of $j$-th layer segments, let $D(B)=(D(S_{1}),\dots, D(S_{|B|}))$. Here, for cohesiveness, we define the segments of layer $1$ as the blocks defined in the block-partition, and define each batch of layer $1$ as a single block. For each node $v\in V$ and each $j\in [d]$, let $S_{v}^{j}$ and $B_{v}^{j}$ be $v$'s segment and batch in layer $j$, respectively. The prover assigns each $v\in V$ with the sequences $D(S_{v}^{j})$ and $D(B_{v}^{j})$ for each $j\in [d]$. In addition, the prover provides proof for the correctness of the $D(S_{v}^{j})$ and $D(B_{v}^{j})$ assignments. Let us first present the verification process assuming that all assigned $D(S_{v}^{j})$ and $D(B_{v}^{j})$ values are correct. Following that, we shall explain how the prover proves their correctness.

Suppose that $B=(S_{1},\dots, S_{|B|})$ and $B'=(S'_{1},\dots, S'_{|B'|})$ are a pair of consecutive batches in layer $j\in \{2,\dots, d\}$ that are placed within the same segment of layer $j-1$. Recall that by the encoding convention, $S_{i}$ and $S'_{i}$ are consecutive segments in an invocation of Lemma \ref{lemma:self-reduction} w.r.t.\ some bitstring. Furthermore, recall that the verification process described in Lemma \ref{lemma:self-reduction} tests equality between pairs of bitstrings assigned to neighboring segments. Based on this notion, we say that segments $S_{i}$ and $S'_{i}$ are \emph{$D$-consistent} if $D(\alpha)=D(\alpha')$ for every pair $\alpha\in A(S), \alpha'\in A(S')$ for which an equality condition is defined. We naturally extend the definition to batches and say that the consecutive batches $B=(S_{1},\dots, S_{|B|})$ and $B'=(S'_{1},\dots, S'_{|B'|})$ are $D$-consistent if every pair of segments $S_{i}$ and $S'_{i}$ are $D$-consistent. For the blocks (i.e., the segments/batches of layer $1$), we define $D$-consistency w.r.t.\ the equality conditions defined by Lemma \ref{lemma:outer-block-to-equality}. That is, we say that a pair $b, b'$ of neighboring blocks in $G$ are $D$-consistent if $D(\alpha)=D(\alpha')$ for the pair $\alpha\in A(b),\alpha'\in A(b')$ for which an equality condition is defined.

The goal of the verifier is to check $D$-consistency for every pair of relevant batches $B,B'$ (as defined above) in every layer $j\in [d]$. The main observation here is that $D$-consistency between $B$ and $B'$ can be checked directly if one is given $D(B)$ and $D(B')$. Indeed, in this case it is just a matter of checking equality between all relevant pairs of bitstrings from $D(B)$ and $D(B')$. Recall that each node $v\in B$ (resp., $v'\in B'$) receives $D(B)$ (resp., $D(B')$) and that the batches $B$ and $B'$ are adjacent. Therefore, there exists a node $v\in B$ neighboring on a node $v'\in B'$. The node $v$ can see $D(B')$ from the label of $v'$ and thus, $v$ can check that $B$ and $B'$ are $D$-consistent.

We go on to explain how the prover proves the correctness of the assignments of $D(S_{v}^{j})$ and $D(B_{v}^{j})$ for each node $v\in V$ and layer $j\in [d]$. To provide the proof, the prover also assigns to each $v\in V$ the position of $S_{v}^{j}$ in the batch $B_{v}^{j}$ for all $j\in [d]$. Moreover, the prover provides $v$ with its position in the $d$-th layer segment $S_{v}^{d}$. Given this assignment, each node $v\in V$ checks that the assignment $D(S_{v}^{j})$ (resp., $D(B_{v}^{j})$) is identical to its segment-neighbors (resp., batch-neighbors). Each node $v\in V$ can also check the correctness of the assigned positions based on standard checks with its segment/batch neighbors. Additionally, $v$ checks that the assignment of $D(B_{v}^{j})$ is compatible with the assignment of $D(S_{v}^{j})$ and the position of $S_{v}^{j}$ in  $B_{v}^{j}$. 

It remains to describe how each node $v\in V$ verifies the correctness of the assignment of the sequences $D(S_{v}^{j})$. First, in the case of $D(S_{v}^{d})$, since $v$ also receives its position in $S_{v}^{d}$, this can be done in a straightforward manner (similarly to the verification of the $A(\mathtt{seg})$ assignment in Section \ref{section:iterated-naive}). For segments of layer $j<d$, we observe that the verifier can use the $D(S_{v}^{j+1})$ assignments to verify the correctness of the $D(S_{v}^{j})$ assignments. More concretely, notice that if a bitstring $\alpha'$ is obtained in layer $j+1$ as a consequence of a $w_{7}(\cdot)$ assignment w.r.t.\ to some bitstring $\alpha$ given in layer $j$, then by definition it should hold that $D(\alpha)=D(\alpha')$. Now, suppose that $S_{v}^{j+1}$ is a segment for which there exists $\alpha'\in A(S_{v}^{j+1})$ that represents a $w_{7}(\cdot)$ assignment w.r.t.\ a bitstring $\alpha\in A(S_{v}^{j})$. Then, $v$ can directly check that the bitstring $D(\alpha)$ in the assignment associated with $D(S_{v}^{j})$ is the same as the bitstring $D(\alpha')$ in the assignment associated with $D(S_{v}^{j+1})$. 

\paragraph{Reducing the number of rounds.} 
The modifications required to reduce the number of interaction rounds are quite simple. Recall that in Section \ref{section:iterated-naive}, the sub-protocol for outer-block edges is invoked recursively within each block resulting in a total of $\Theta(\log ^{*}n)$ invocations (for a total $\Theta((\log^{*}n)^{2})$ interaction rounds in the protocol). Here, we observe that the $j$-th recursive invocations within the blocks can be implemented \emph{in parallel} to the assignments of the $j$-th layer in the sub-protocol for outer-block edges. That is, we augment the sub-protocol for outer-block edges such that when the prover assigns the $j$-th layer, it also produces the block partition and label assignment that corresponds to the $j$-th recursive invocation of the sub-protocol. By construction, this only adds a constant number of equality instances for each segment and thus, does not (asymptotically) affect the proof size.

\paragraph{Complexity.}
As discussed above, the number of rounds in the modified protocol is $O(d)=O(\log ^{*}n)$. Regarding proof size, first notice that in each round during the construction of the $d$ layers, each node receives a constant number of bits, i.e., the construction of layers admits a constant proof size. 

Let us now consider the components that make up the label assignment after the layer construction. First, each node $v\in V$ is labeled with its position in $S_{v}^{d}$. This can be implemented with constant-size labels since the $d$-th layer segments are of constant size. Let us now consider the label assignment associated with layer $j\in [d]$. Each node $v\in V$ receives the sequences $D(S_{v}^{j})$ and $D(B_{v}^{j})$ as well as the position of $S_{v}^{j}$ in $B_{v}^{j}$. For the position assignment, recall that each batch consists of a constant number of segments. Therefore, the position only requires constant size. As for $D(S_{v}^{j})$ and $D(B_{v}^{j})$, first notice that by construction, every $D(\alpha)$ is of constant size (as it is a bitstring assigned in the $d$-th layer). Moreover, the encoding scheme of the layers is defined such that every segment $S_{v}^{j}$ receives a constant number of bitstrings and every batch $B_{v}^{j}$ consists of a constant number of segments. Therefore, for every $j\in [d]$, the sequences $D(S_{v}^{j})$ and $D(B_{v}^{j})$ only require a constant number of bits. This gives a total of $O(d)$ communication needed from the prover. Recall that to get a constant proof size, the prover can simply break the labels into $d$ constant-sized pieces and send them to the verifier over $d$ interaction rounds. Overall, this results in the desired $O(d)$ interaction rounds and $O(1)$ proof size.

\paragraph{Correctness.}
The completeness of the protocol is similar to the completeness presented in Section \ref{section:iterated-naive}. We now analyze the soundness. For a pair $\alpha,\alpha'$ of bitstrings assigned to some segment in layer $j\in [d]$, we say that $\alpha$ and $\alpha'$ are \emph{mutually constrained} if the event $D(\alpha)\neq D(\alpha')$ implies that the verifier rejects the instance. Note that for $\alpha$ and $\alpha'$ to be mutually constrained, $D(\alpha)$ and $D(\alpha')$ do not have to be compared directly by the verifier as part of the protocol. Indeed, it could be the case that $\alpha$ and $\alpha'$ are mutually constrained due to transitivity. To prove the soundness of the protocol, we now prove the following lemma.

\begin{lemma}\label{lemma:constrained}
	Let $p_{1},p_{2},\dots, p_{d}$ be a sequence of real values defined such that $p_{d}=0$ and $p_{j}=p_{j+1}+1/\log ^{(j)}n$ for all $1\leq j\leq d-1$. Then, for all $j\in [d]$, if the $j$-th layer admits mutually constrained bitstrings $\alpha,\alpha'$ such that $\alpha\neq\alpha'$, then the verifier rejects the instance with probability at least $1-p_{j}$.
\end{lemma}
Before proving the lemma we note that it immediately implies the soundness of the protocol. To see that, first recall that by the soundness analysis of Section \ref{section:iterated-naive}, it follows that if the prover produces a dishonorable assignment at any time during the protocol, then the verifier rejects with probability at least $7/8$. In addition, soundness analysis of Section \ref{section:iterated-naive} implies that any honorable assignment must admit a mutually constrained pair $\alpha\neq \alpha'$ in some layer $j\in[d]$. In this case, Lemma \ref{lemma:constrained} implies that the verifier rejects with probability at least $1-p_{j}$. The soundness follows since $1-p_{j}\geq 7/8$ (as established in the soundness analysis of Section \ref{section:iterated-naive}).

\begin{proof}[Proof of Lemma \ref{lemma:constrained}]
	We prove the lemma by induction on the layers in reverse order. For the base case, note that if $\alpha$ and $\alpha'$ are assigned in the $d$-th layer, then by definition $D(\alpha)=\alpha,D(\alpha')=\alpha'$. Therefore, $\alpha\neq\alpha'$ implies that $D(\alpha)\neq D(\alpha')$ and since $\alpha$ and $\alpha'$ are mutually constrained, it follows that the verifier rejects the instance with probability $1$.
	
	We now consider layer $j<d$ and assume that the lemma holds for all layers $j'>j$. Recall that the $w_{7}(\cdot)$ assignments for $\alpha$ and $\alpha'$ correspond to the polynomial values $\Phi_{q}(r;\alpha)$ and $\Phi_{q}(r;\alpha')$, respectively. Since $\alpha\neq \alpha'$, we get that the probability of $\Phi_{q}(r;\alpha)=\Phi_{q}(r;\alpha')$ is at most $1/\log ^{(j)}n$. Let us condition on the event that $\Phi_{q}(r;\alpha)\neq \Phi_{q}(r;\alpha')$. Let $\beta$ and $\beta'$ be bitstrings assigned in the $(j+1)$-th layer as an encoding of the $w_{7}(\cdot)$ assignment associated with $\alpha$ and $\alpha'$, respectively. If the prover assigns $\beta=\beta'$, then this must mean that either $\beta\neq \Phi_{q}(r;\alpha)$ or $\beta'\neq \Phi_{q}(r;\alpha')$. In either case, it follows from the soundness analysis presented in Section \ref{section:self-reduction} that there exist two mutually constrained bitstrings $\gamma,\gamma'$ assigned in the $(j+1)$-th layer such that $\gamma\neq \gamma'$. Now consider the case where $\beta\neq \beta'$. Notice that by definition, it holds that $D(\beta)=D(\alpha)$ and $D(\beta')=D(\alpha')$. Since $\alpha$ and $\alpha'$ are mutually constrained, so are $\beta$ and $\beta'$. To conclude, we get that the probability that all pairs of mutually constrained bitstrings in the $(j+1)$-th layer are equal is bounded from above by $1/\log ^{(j)}n$. By the induction hypothesis, this means that the verifier rejects the instance with probability at least $1-1/\log ^{(j)}n-p_{j+1}=1-p_j$.
\end{proof}

%% file: appendix.tex
\appendix
\begin{figure*}[!t]
	{\centering
		\Large{APPENDIX}
		\par}
\end{figure*}

\section{Missing Proofs of Section \ref{section:preliminaries}}\label{section:missing-proofs-prelims}
\begin{proof}[Proof of Lemma \ref{lemma:greater-than}]
	Given a path $P$ and integers $\alpha,\beta$, we describe a non-interactive protocol for GT. Let $i$ be the index of the most significant bit in which $\alpha$ and $\beta$ differ. Let $v$ be the node that receives $\alpha[i]$ and $\beta[i]$ in the encoding. For each node $u\in V$, the prover assigns the value $0$ if $u$ is to the left of $v$ in $P$, and $1$ if $u=v$.
	
	Given the label assignment, each node $u\in V$ operates as follows. If $u$ received the value $0$ from the prover, then it checks that $\alpha(u)=\beta(u)$ and that its left neighbor (if such exists) received the value $0$ as well. If $u$ received the value $1$ from the prover, then it checks that  $\alpha(u)=1,\beta(u)=0$ and that its left neighbor received the value $0$ from the prover. The correctness of the protocol follows from the fact that $\alpha>\beta$ if and only if $\alpha[i]>\beta[i]$ and $\alpha[i']=\beta[i']$ for all $i'>i$.  
\end{proof}

\begin{proof}[Proof of Lemma \ref{lemma:addition}]
	For every $i\in \{0,\dots, \ell\}$, define $\mathtt{carry}(i)$ recursively such that $\mathtt{carry}(0)=0$ and $\mathtt{carry}(i)=\mathds{1}(\alpha[i]+\beta[i]+\mathtt{carry}(i-1)\geq 2)$. That is, $\mathtt{carry}(i)$ represents the carry bit associated with the $i$-th least significant bit in a standard addition process for each $i\in [\ell]$. To implement the protocol, the prover simply assigns $\mathtt{carry}(i)$ to the $i$-th rightmost node for each $i\in [\ell]$. This assignment allows the nodes to verify that indeed $\alpha+\beta=\gamma$.
\end{proof}

\begin{proof}[Proof of Lemma \ref{lemma:equality}]
	Let $q$ be a prime number in the range $[\ell^{2},2\ell^{2}]$ which is known in advance to all nodes of $P$ and $P'$. For each node $v\in P$ (resp., $v\in P'$), let $i_{v}\in [\ell]$ denote $v$'s position in $P$ (resp., $P'$), i.e., $v$ is the $i_{v}$ leftmost node of $P$ (resp., $P'$). Let us denote $bit_{v}=\alpha(v)$ if $v\in P$ and $bit_{v}=\alpha'(v)$ if $v\in P'$.
	
	The protocol begins with the verifier at the leftmost node of $P$ drawing a value $r\in \{0,\dots , q-1\}$ u.a.r.\ and sending it to the prover. Then, for each node $v\in P$ (resp., $v\in P'$), the prover assigns the following values: $z_{1}(v)=r$; $z_{2}(v)=i_{v}$; $z_{3}(v)=\Phi_{q}(r;\alpha[1,\dots, i_{v}])$; and $z_{4}(v)= \Phi_{q}(r;\alpha)$. Notice that if $v\in P'$, then the assignment of the honest prover satisfies $z_{3}(v)=\Phi_{q}(r;\alpha'[1,\dots, i_{v}])=\Phi_{q}(r;\alpha[1,\dots, i_{v}])$ and $z_{4}(v)= \Phi_{q}(r;\alpha')= \Phi_{q}(r;\alpha)$.
	
	Let us now describe the verification process. First, the leftmost node $v\in P$ verifies that $z_{1}(v)=r$ and rejects otherwise. Now, consider some node $v\in P$ (resp., $v\in P'$) which is not the leftmost or rightmost node. Let $u_{1}$ and $u_{2}$ denote $v$'s left and right path-neighbors, respectively. The verifier at $v$ verifies that $z_{1}(v)=z_{1}(u_{1})=z_{1}(u_{2})$ and $z_{4}(v)=z_{4}(u_{1})=z_{4}(u_{2})$. In addition, the verifier checks that $z_{2}(v)=z_{2}(u_{1})+1=z_{2}(u_{2})-1$ and that $z_{3}(v)=(z_{3}(u_{1})+bit_{v}\cdot z_{1}(v)^{z_{2}(v)})\bmod q$. If $v$ is the leftmost node of $P$ or $P'$, then the verifier checks that $z_{2}(v)=1$ and $z_{3}(v)=bit_{v} \cdot z_{1}(v)^{z_{2}(v)} \bmod q$. If $v$ is the rightmost node of $P$ or $P'$, then the verifier checks that $z_{3}(v)=z_{4}(v)=(z_{3}(u)+bit_{v}\cdot z_{1}(v)^{z_{2}(v)})\bmod q$ where $u$ is the left path-neighbor of $v$. Finally, $u^{*}$ and $v^{*}$ verify that $z_{1}(u^{*})=z_{1}(v^{*})$ and $z_{4}(u^{*})=z_{4}(v^{*})$.
	
	The proof size of the protocol is $O(\log \ell)$ since each $z_{i}(v)$ value can be encoded using $O(\log \ell)$ bits. For correctness, we first note that the protocol has perfect completeness. Indeed, if $\alpha=\alpha'$, then $\Phi_{q}(r;\alpha)= \Phi_{q}(r;\alpha')$ and the label assignment produced by the honest prover satisfies all verification conditions. As for soundness, first notice that if all verification conditions above are satisfied, then $z_{4}(v)=\Phi_{q}(r;\alpha)$ for all $v\in P$ and $z_{4}(v)= \Phi_{q}(r;\alpha')$  for all $v\in P'$. This means that the condition $z_{4}(u^{*})=z_{4}(v^{*})$ can hold if and only if $\Phi_{q}(r;\alpha)= \Phi_{q}(r;\alpha')$. By Lemma \ref{lemma:poly-identity}, in the case that $\alpha\neq \alpha'$, this event occurs with probability at most $\ell/q<\ell/\ell^{2}=1/\ell$.
\end{proof}

\begin{proof}[proof of Lemma \ref{lemma:multiplication-to-equality}]
	The idea of the protocol is to follow the basic process for multiplication between two numbers. To that end, the prover divides $P$ into $\ell'$ clusters each of which consists of $2\ell'$ consecutive nodes in $P$. Let $P(i)$ denote the $i$-th cluster for each $i\in [\ell']$. For each $i\in [\ell]$, let $\beta_{R}[i]$ denote the $i$-th \emph{least significant} bit of $\beta$ and let $\beta_{R}[1,\dots, i]$ be the 	number represented by the $i$ least significant bits of $\beta$. 
	Let $\beta_{\text{po$2$}}(i)$ denote the number $\beta_{R}[i]\cdot 2^{i-1}$.
	
	For each $1<i\leq \ell'$, the prover assigns $P(i)$ with the following values: $z_{1}(i)=\alpha$, encoded such that it is shifted $i-1$ places to the left; $z_{2}(i)=\beta$; $z_{3}(i)=\gamma$; $z_{4}(i)=\beta_{R}[1,\dots, i]$; $z_{5}(i)=\beta_{R}[1,\dots, i-1]$; $z_{6}(i)=\beta_{\text{po$2$}}(i)$; $z_{7}(i)=\alpha\cdot\beta_{\text{po$2$}}(i)$; $z_{8}(i)=\alpha\cdot \beta_{R}[1,\dots, i-1]$; and $z_{9}(i)=\alpha\cdot \beta_{R}[1,\dots, i]$. In addition, the prover provides proof that $z_{9}(i)=z_{7}(i)+z_{8}(i)$ for each $i\in[\ell']$ based on the protocol of Lemma \ref{lemma:addition}. Given this assignment, it is straightforward for the verifier to check that $z_{4}(i),z_{5}(i),z_{6}(i)$ are consistent with each other for each $i\in [\ell']$. Moreover, due to the shifting in the encoding of $z_{1}$, the verifier can check that $z_{7}(i)=z_{1}(i)\cdot z_{6}(i)$ in a direct manner. The verifier in the last cluster $P(\ell')$ also checks that $z_{3}(\ell')=z_{9}(\ell')$.
	
	To verify the correctness of the assignment, the verifier seeks to check the following equalities for each $1<i\leq \ell'$: (1) $z_{j}(i-1)=z_{j}(i)$ for every $j\in\{1,2,3\}$; (2) $z_{5}(i)=z_{4}(i-1)$; (3) $z_{8}(i)=z_{9}(i-1)$. To verify these equalities, the equality protocol $\mathsf{EQ}_{2\ell'}$ is applied between every pair $P(i-1),P(i)$ in parallel.
\end{proof}

\section{The Necessity of Interaction Rounds}\label{section:interaction}
We note that in the $\iterated$ protocol of Section \ref{section:recursion}, we invoke the self-reductions one after the other, resulting in $\Theta(\log ^{*}n)$ rounds. A question that may arise is why can't the verifier generate all its randomness in a single round and then have the prover provide the entire transcript of the protocol in response to that randomness? The short answer to this question is that a cheating prover can take advantage of knowing the randomness in advance and engineer a false proof. 

To make things more concrete, let us look at an attempt to achieve this round-reduction for two invocations of the self-reduction (a larger number of invocation can only help the prover). Suppose that we have a no-instance, i.e., there exists an equality instance $(\alpha,\alpha')$ such that $\alpha\neq \alpha'$. Let $r_{1}\in \{0,\dots, q_{1}-1\},r_{2}\in \{0,\dots, q_{2}-1\}$ be the random choices of the verifier where $q_{1}=\poly(|\alpha|)$ and $q_{2}=\poly\log |\alpha|$ are the primes chosen in accordance with the self-reduction. To get the verifier to accept, the prover can employ the following strategy: look for $r'\in \{0,\dots, q_{1}-1\}$ such that: (1) $\Phi_{q_{1}}(r';\alpha')=\Phi_{q_{1}}(r_{1};\alpha)$, and (2) $\Phi_{q_{2}}(r_{2};r')=\Phi_{q_{1}}(r_{2};r_{1})$; if such $r'$ exists, encode it to the label associated with the proof of the first interaction for $\alpha'$. Notice that such label assignment would cause the verifier to accept. Furthermore, the number of $r'$ values that satisfy the first equation can be as large as the degree of $\Phi_{q_{1}}(;\alpha')$, i.e., $|\alpha'|$. The prover can go over all such $r'$ values and try to find one that satisfies the second equation. Since $|\alpha'|$ is exponentially larger than $q_{2}$, it is very likely that such $r'$ exists.  

\section{Missing Proofs of Section \ref{section:double-log}}\label{section:missing-proofs-double}
\begin{proof}[Proof of Lemma \ref{lemma:gp-soundness}]
	If the label assignment is dishonorable, then by definition there is an outer-block edge $e$ whose endpoints belong to different blocks $b,b'$. In this case, it follows from the construction that the verifier rejects unless $r_{b}=r_{b'}$ which occurs with probability $1/2^{\log \log n}=1/\log n$. So, assume that the label assignment given by the prover is honorable. This means that any edge that was labeled as inner-block is in fact an inner-block edge. Therefore, if there is a violating edge $e=(u,v)$ that was labeled as inner-block, then $\mathtt{index}(v)<\mathtt{index}(u)$ and the verifier rejects. 
\end{proof}

%%%%%%%%%%%%%%%%%%%%%%%%%%%%%%%%%%%%%%%%%%%%%%%%%%%%%%%%%%%%%%%%%%%%%%%%%%%%%%
\section{Lifting the Simplifying Assumptions of Section \ref{section:simplified}}\label{section:lifting}
%%%%%%%%%%%%%%%%%%%%%%%%%%%%%%%%%%%%%%%%%%%%%%%%%%%%%%%%%%%%%%%%%%%%%%%%%%%%%%
In this section, we explain how to design the protocol of Section \ref{section:simplified} without making the simplifying assumptions.

\paragraph{The case of $g>1$:} Extending the protocols to values $g>1$ is simple. Instead of writing bitstrings on segments of $c_{1}\log \ell$ nodes as described in Section \ref{section:simplified}, the prover instead writes each bitstring on a segment of $g\cdot c_{1}\log \ell$ nodes where the $i$-th bit of the string is assigned to the $((i-1)g+1)$-th leftmost node in the segment. This turns every instance of $(1,c_{1}\log \ell)$-equality in Section \ref{section:simplified} into an instance of $(g,c_{1}\log \ell)$-equality which is in accordance with Lemma \ref{lemma:self-reduction}.

\paragraph{Lifting the assumption of $r$'s encoding:} We now lift the assumption that the random variable $r$ is initially written on the $c_{1}\log \ell$ leftmost nodes of $P$ and $P'$. Notice that here we still assume that the prime number $q$ is encoded on these nodes (we explain how to lift this assumption below). Furthermore, assume that $q$ is encoded on $P$ such that all nodes that are not assigned a leading zero are \emph{marked}. Notice that this assumption can be obtained trivially by having the prover assign an additional bit to the labels of the nodes that are used in $q$'s representation. 

A variable $r\in\{0,\dots ,2^{\lfloor\log q\rfloor}-1\}$ is then sampled by having every marked node apart from the leftmost draw a random bit. Notice that $r$ is sampled only by nodes in $P$. Then, the value of $r$ is given to all clusters in $P'$ through the label assignment of the prover. To check that the prover assigned correct $r$ values to the clusters of $P'$, the equality protocol $\Pi$ is executed similarly to the final verification step presented in the protocol of Section \ref{section:simplified}. 

Observe that $r$ is sampled u.a.r.\ from the range $\{0,\dots ,2^{\lfloor\log q\rfloor}-1\}$ rather than the range $\{0,\dots ,q-1\}$ assumed in Section \ref{section:simplified}. We argue that given this sampling process, the soundness error can be bounded by $\soundness(\Pi)+1/\ell$. To see that, suppose that $\alpha\neq \alpha'$. The probability that $\Phi_{q}(r,\alpha)=\Phi_{q}(r,\alpha')$ is equal to the probability that $r$ is a root of the (non-zero) polynomial $\Phi_{q}(\cdot,\alpha)-\Phi_{q}(\cdot,\alpha')$. Since the polynomial is of degree at most $\ell$, it follows that this probability is bounded from above by $\ell/2^{\lfloor\log q\rfloor}<\ell/(q/2)<\ell/\ell^{2}=1/\ell$. The rest of the correctness and complexity analyses remain unaffected.

\paragraph{Lifting the assumption of $q$'s encoding:} The assumption on $q$'s encoding is lifted by having the prover encode $q$ to the leftmost nodes of $P$ and $P'$. Then, the prover seeks to prove the following: (1) $q$ is in the range $[2\ell^{2},3\ell^{2}]$; (2) $q$ is a prime number; and (3) $P$ and $P'$ received the same value $q$. The verification of (3) can be obtained by means of an execution of the equality protocol $\Pi$, similarly to the way $\Pi$ is used above to verify the correctness of the $r$ values assigned to $P'$. 

Recall that as part of the procedure described in Section \ref{section:simplified}, the prover assigns $q$ to every segment $S$. Let $Q(S)$ denote the nodes of $S$ that participate in the encoding of $q$ (i.e., not leading zeros). Furthermore, recall that the equality protocol $\Pi$ is used to check that $Q(S)$ encode the same value for each segment $S$. Thus, we assume w.l.o.g.\ that the size $|Q(S)|$ is identical for every segment $S$ and denote this size by $x$. To provide proof of (1), it is sufficient for the prover to show that either $x=\log (2\ell^{2})=2\log \ell+1$ or $x=\lceil \log (3\ell^{2})\rceil=2\log \ell+2$. We present the proof for the case $x=2\log \ell+1$; it is simple to derive the proof for the case $x=2\log \ell+2$.

For every segment $S$ of cluster $i\in [\ell]$, the prover assigns the values $z_{1}(i)=i-1$ and $z_{2}(i)=i-2$ to $Q(S)$ (for cluster $1$, no $z_{2}$ value is given). The values are encoded such that the bits of $z_{1}(i)$ and $z_{2}(i)$ are assigned only to nodes at even positions of $Q(S)$. The prover also provides a proof that $z_{1}(i)=z_{2}(i)+1$ for every segment of cluster $i$ based on Lemma \ref{lemma:addition}. 

For the verification, each pair of adjacent segments in cluster $i$ use $\Pi$ to check that they received the same assignments $z_{1}(i),z_{2}(i)$. In addition, the rightmost segment of cluster $i-1$ and the leftmost segment of cluster $i$ use $\Pi$ to check that $z_{1}(i-1)=z_{2}(i)$ for every $2\leq i\leq \ell$. Finally, the verifier at each segment of the last cluster $\ell$ checks that the leftmost node that participates in the encoding of $z_{1}(\ell)$ receives the bit $1$. This completes the proof for (1). Indeed, if $x<2\log \ell+1$, then the values $z_{1}(i)$ are encoded on fewer than $\log \ell$ nodes which means that they cannot encode all values in $\{0,\dots, \ell-1\}$. On the other hand, if $x>2\log \ell+1$, then an assignment of $z_{1}(\ell)=\ell-1$ would assign the leftmost participating node with the bit $0$. In either case, we get that any assignment of $z_{1}(i),z_{2}(i)$ values violates some verification conditions and therefore fails with probability at least $1-\soundness(\Pi)$.

We now describe the proof of (2). For each segment of cluster $i\in [\ell]$, the prover assigns the following values: $y_{1}(i)=2i-1$, i.e, $y_{1}(i)$ is the $i$-th odd number; $y_{2}(i)=y_{1}(i-1)=2(i-1)-1$; and $y_{3}(i)$ defined as the inverse of $2i-1$ in $\mathbb{F}_{q}$. The prover also provides proof that $y_{1}(i)=y_{2}(i)+2$ based on Lemma \ref{lemma:addition} for each segment of cluster $i$. Moreover, the prover proves to each cluster $i\in [\ell]$ that $y_{1}(i)\cdot y_{3}(i)\equiv 1 \bmod q$. This is done by reducing modular multiplication to equality tasks between adjacent segments based on Lemmas \ref{lemma:multiplication-mod-p} and \ref{lemma:multiplication-to-equality}. Then, $\Pi$ is used to solve each equality task between adjacent segments. Finally, similarly to before, $\Pi$ is used to check that all segments of cluster $i$ received the same assignment and that $y_{1}(i-1)=y_{2}(i)$ for every $2\leq i\leq \ell$. In addition, the verifier at every segment checks that its assigned value corresponding to $q$ is an odd number. This is done simply by checking that the rightmost node receives the bit $1$ in the $q$-assignment.

We argue that the described assignment provides a proof that the assigned value $q$ is a prime number. To see why, recall that for $q$ to be prime it is sufficient that $q$ is not divisible by any integer $1\leq a\leq \sqrt{q}$. In the described scheme, it is checked that $q$ is not divisible by any even number and that $q$ is not divisible by any odd number which is smaller than $2\ell$. To complete the argument, recall that $2\ell>\sqrt{3}\ell>\sqrt{q}$.